\documentclass[journal]{IEEEtran}
\usepackage{graphicx}
\usepackage{subfigure}
\usepackage{cite}
\usepackage{bm}
\usepackage{amsmath}   
\usepackage{amsthm}
\usepackage{cases}
\usepackage{amsfonts,amssymb}
\usepackage{cuted}\stripsep-3pt
\usepackage{lipsum}

\usepackage[absolute,showboxes]{textpos}

\TPMargin{3pt}

\newcommand{\copyrightstatement}{
    \begin{textblock}{0.84}(0.08,0.95) 
         \noindent
         \footnotesize
         \copyright 2021 IEEE. Personal use of this material is permitted. Permission from IEEE must be obtained for all other uses, in any current or future media, including reprinting/republishing this material for advertising or promotional purposes, creating new collective works, for resale or redistribution to servers or lists, or reuse of any copyrighted component of this work in other works. DOI: 10.1109/TCOMM.2021.3110871.
    \end{textblock}
}

\newtheorem{theorem}{{Theorem}}

\ifCLASSINFOpdf
\else
\fi
\begin{document}
\copyrightstatement
\title{Modeling and Performance Analysis\\ of OAM-NFC Systems}


\author{\IEEEauthorblockN{Runyu Lyu, \textit{Student Member, IEEE}, Wenchi Cheng, \textit{Senior Member, IEEE}, and Wei Zhang, \textit{Fellow, IEEE}~}

\thanks{Part of this work was presented in IEEE International Conference on Communications, 2020~\cite{OAM_NFC_ICC}.

This work was supported by the National Natural Science Foundation of China (No. 61771368), Foundation of CETC Key Laboratory of Data Link Technology (CLDL-20182411), Key Area Research and Development Program of Guangdong Province under grant No. 2020B0101110003, and in part by Shenzhen Science \& Innovation Fund under Grant JCYJ20180507182451820.

Runyu Lyu and Wenchi Cheng are with Xidian University, Xi'an, 710071, China (e-mails: rylv@stu.xidian.edu.cn; wccheng@xidian.edu.cn).

Wei Zhang is with the School of Electrical Engineering and Telecommunications, the University of New South Wales, Sydney, Australia (e-mail: w.zhang@unsw.edu.au)
}

%
%
}

%



\IEEEtitleabstractindextext{%
\begin{abstract}
Due to its low energy consumption and simplicity, near field communication (NFC) has been extensively used in various short-range transmission scenarios, for example, proximity payment and NFC entrance guard. However, the low data rate of NFC limits its application in high rate demanded scenarios, such as high-resolution fingerprint identification and streaming media transmission as well as the future promising high rate indoor communications among pads, phones, and laptops. In this paper, we model and analyze the performance of the orbital angular momentum based NFC (OAM-NFC) system, which can significantly increase the capacity of NFC. We first give the OAM system model. With coils circularly equipped at the transmitter and receiver, OAM-NFC signals can be transmitted, received, and detected. Then, we develop the OAM-NFC generation and detection schemes for NFC multiplexing transmission. We also analyze the OAM-NFC channel capacity and compare it with those of single-input-single-output (SISO) as well as multi-input-multi-output (MIMO) NFC. Simulation results validate the feasibility and capacity enhancement of our proposed OAM-NFC system. How different variables, such as the transceiver misalignment, the numbers of transceiver coils, and transceiver distance, impact the OAM-NFC capacity are also analyzed.
\end{abstract}

\begin{IEEEkeywords}
Near field communication (NFC), orbital angular momentum (OAM), high-data-rate short-range transmission.
\end{IEEEkeywords}}

\maketitle

\IEEEdisplaynontitleabstractindextext

%
\IEEEpeerreviewmaketitle

\section{Introduction}
\IEEEPARstart{N}{ear} field communication (NFC), which is based on the near field magnetic induction\cite{NFMC} at a distance much less than a wavelength of the carrier, is a kind of simple and efficient short-range wireless communication\cite{NFC,NFC_overview}. Devices with NFC can achieve contactless, point-to-point, and low-power-consumption data transmission in a very short range\cite{NFC_cellphone}. The short communication range facilitates NFC against various attacks, such as packet capturing, malicious eavesdropping, and wireless interference. Thus, many security-sensitive applications, such as proximity payment, magnetic stripe ID card, and NFC entrance guard, rely on NFC technology\cite{NFC_security}. Moreover, different from the conventional radio frequency (RF) communications, which are based on electromagnetic (EM) waves, NFC is much more reliable in high permittivity channels, such as underwater or underground channels.

According to ISO/IEC 18092, the data rate of NFC is $106$, $212$, or $424$ kbps, while the maximum data rate of NFC is $848$ kbps\cite{NFC_data_rate}. This data rate can support conventional NFC applications, such as the proximity payment and NFC entrance guard, where a high data rate is not required. However, with the emerging of diverse mobile data transmission applications, it is highly demanded but very challenging to meet the capacity enhancement requirement \cite{What_will_5G_be}. The current low data rate of NFC limits its application in future high rate demanded short-range contactless data transmission scenarios, such as the streaming media transmission and high-resolution fingerprint information identification with some complex security codes\cite{Trend_of_6G,NFC_Capacity}. Also, the low data rate prevents us from taking advantage of NFC's reliability under high permittivity channels, such as underground soil medium, oil reservoirs, and mines.

Some schemes can increase the capacity for NFC systems. For instance, high-data-rate NFC can be achieved by combining NFC with Bluetooth or WiFi, which uses NFC to establish a link and utilizes Bluetooth or WiFi to transfer data. However, the relatively low security and the RF interference make this combination not a well-applied solution\cite{NFMC}. Another scheme to increase the capacity for NFC is integrating multi-input-multi-output (MIMO) technology into conventional single-input-single-output (SISO) NFC system\cite{MIMO_NFC}. However, the crosstalks caused by the inductive coupling among transmit coils can severely reduce the capacity of MIMO-NFC system\cite{MIMO_NFC_nocross}. Also, NFC technology is often used for line of sight (LOS) scenarios, which are with high channel correlation and low rank thus not fitting for MIMO-based schemes.

Using orbital angular momentum (OAM) carried EM waves to transmit information is a novel way that can increase the capacity without additional power and bandwidth for RF transmissions\cite{oam_low_freq_radio,oam_for_wireless_communication,oam_mode_division,DBLP:journals/jcin/LiangLCZ20}. Thanks to the orthogonal wavefronts of OAM based EM waves, many researchers have been considering OAM\cite{oam_light} as a tool for multiplexing signals transmission in the last decade. Studies and experiments validate that OAM technology can be employed in practical RF wireless communication systems and increase the capacity\cite{oam_generation_detection,oam_mode_modu,oam_multiplexing,oam_embedded_massive_mimo,oam_mimo_capacity,oam_deeplearing,oam_planespiral,DBLP:journals/jcin/GaoC019}. Moreover, because of the strict wavefront limitation, OAM based wireless communication is preferred to be used for LOS scenarios\cite{oam_hollow,oam_missaligned,oam_direct_generation}, such as data center wireless communication. Since the capacity enhancement achieved by OAM is not based on multipath but the orthogonal wavefronts, the strong correlation of the channel does not decrease the capacity. Thus, it is highly promising to use OAM for capacity enhancement in NFC scenarios as in RF scenarios\cite{oam_MIMO}.

However, unlike RF communication that relies on electromagnetic wave propagation, NFC is achieved by magnetic field induction between transmit and receive coils at distances much less than a wavelength of the carrier (NFC is used at distances under $10$ cm while with a wavelength of $22$ m), which makes the NFC channel completely different from the RF channel and makes it not known whether OAM can be used in NFC. Also, whether it can bring capacity enhancement for NFC has not been proven. To evaluate the feasibility of OAM-NFC transmission scheme and its capacity enhancement, in this paper, we model and analyze the performance of the OAM based NFC (OAM-NFC) system. Based on the OAM-NFC system model, we derive the OAM-NFC mutual inductance channel matrix. The difference between RF wireless transmission and NFC is that RF transmission is based on EM wave propagation at a distance farther than a wavelength of the carrier, while NFC is based on the magnetic induction between each transmit-receive pair coils at a distance much less than a wavelength. This difference makes it difficult to analyze the multi-coil mutual inductance and model the multi-coil channel as multi-antenna channels in RF transmission, which makes modeling the OAM-NFC mutual inductance channel one of our major theoretical breakthroughs. Furthermore, based on the channel matrix, the OAM-NFC generation and detection schemes are proposed. The crosstalks caused by inductive coupling among transmit coils are validated not a barrier for the recovery of the input signals when the transmit and the receive coil rings align with each other. Finally, we validate the feasibility and capacity enhancement of our proposed OAM-NFC system with simulations. Also, we analyze how different variables, such as the numbers of transceiver coils and transceiver distance, impact the channel capacity.

The rest of the paper is organized as follows. Section~\ref{sec:System_model} gives the OAM-NFC system model. In Section~\ref{sec:Scheme}, the OAM-NFC generation and detection schemes are proposed. Section~\ref{sec:Capacity} gives the OAM-NFC capacity and analyses. Channel capacities of SISO-NFC and MIMO-NFC are also given for comparison. In Section~\ref{sec:Simulation}, we give simulation results that validate the feasibility and capacity enhancement of the OAM-NFC system. The impacts of the transceiver misalignment, coil numbers, transceiver distance, transceiver coil ring radii, and coil radii on the OAM-NFC capacity are also evaluated. This paper concludes with Section~\ref{sec:Conclusion}.

\section{System Model}\label{sec:System_model}
\begin{figure}[htbp]
\centering
\includegraphics[scale=0.45]{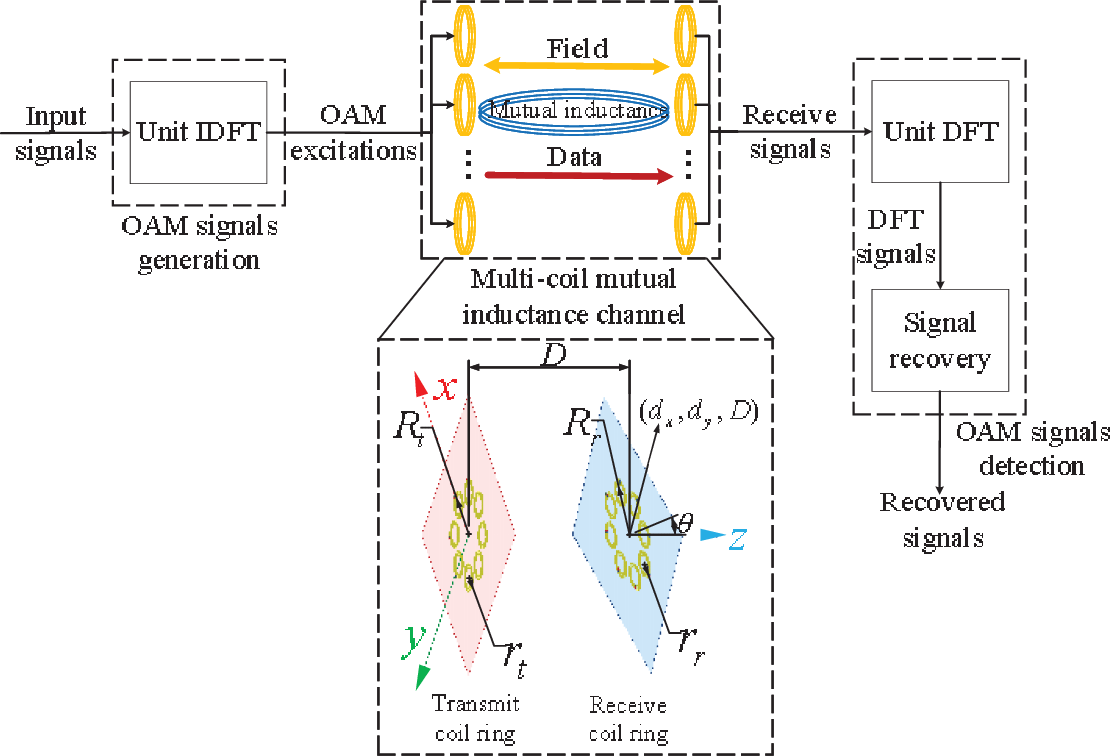}
\caption{The OAM-NFC system model.} \label{fig:system_model}
\end{figure}
Figure~\ref{fig:system_model} shows our proposed OAM-NFC system model, where the transmit and receive antennas are coil rings with $N_t$ and $N_r$ coils, respectively. We chose coil rings as the transceivers due to their flexibility and controllability of digitally generating beams with multiple OAM modes simultaneously. The input signals are first transformed to the OAM excitations on transmit coils using unit inverse discrete Fourier transform (IDFT). This transform is equal to simultaneously feeding the equidistantly distributed transmit coils around the transmit ring by equal-amplitude excitations with linearly increased phases. The phase of adjacent coils linear increases by $2\pi l/N_t$, where $l$ is the index of OAM mode ($0\hspace{-0.1cm}\le\hspace{-0.1cm}l\hspace{-0.1cm}\le\hspace{-0.1cm}N_t\hspace{-0.1cm}-\hspace{-0.1cm}1$). Then, the received induced signals are sent to the signal recovery after unit discrete Fourier transform (DFT).

Figure~\ref{fig:system_model} also shows the multi-coil inductance channel. The centers of the transmit and the receive coils formulate the transmit uniform ring and the receive uniform ring, respectively. All the normals of the transmit and the receive coils are parallel, respectively. We denote by $R_t$ and $R_r$ the radii of transmit and receive coil rings, respectively. The radii of transmit and receive coils are denoted by $r_t$ and $r_r$, respectively. We also denote by $K_t$ and $K_r$ the numbers of turns for each transmit and each receive coil. The coordinate system is defined as follows. The $x$-axis is set from the center of the transmit coil ring to the center of one of the transmit coils; the $z$-axis is set along the axis of the transmit coil ring towards the transmission direction; the $y$-axis is given following the right-hand spiral rule. The center of transmit coil ring is located at the coordinate origin while $(d_x,d_y,D)$ denotes the location of receive coil ring center. $\theta_x$ and $\theta_y$($0\hspace{-0.1cm}\le\hspace{-0.1cm}\theta_x, \theta_y\hspace{-0.1cm}\le\hspace{-0.1cm}\pi/2$) denote the angle between the transmit and the receive coil ring normal lines along $x$ and $y$ axes, respectively. The deflection angle between the transmit and the receive coil ring normal lines, denoted by $\theta$, can be given by $\theta\hspace{-0.1cm}=\hspace{-0.1cm}\arctan\sqrt{\tan^2\theta_x+\tan^2\theta_y}$. We number the transmit coil, of which the center is in the positive $x$-axis, the $1$st transmit coil. The receive coil, which is opposite to the $1$st transmit coil, is set as the $1$st receive coil. The rest transmit and receive coils are numbered from $2$ to $N_t$ and $N_r$ in a clockwise order viewed along $z$-axis, respectively.

\section{OAM-NFC Generation and Detection Schemes}\label{sec:Scheme}
In this section, we propose the OAM-NFC generation and detection schemes without and with channel estimation. Then, we show the feasibility of our proposed schemes.

\begin{figure*}[tp]
\setcounter{equation}{3}
\begin{subequations}
\begin{numcases}{}
M_{m,n}=\frac{\mu_0K_tK_rr_tr_r}{4\pi}\hspace{-0.2cm}\int_0^{2\pi}\hspace{-0.2cm}\int_0^{2\pi}
\Big[r_t^2+r_r^2+d_{m,n}^2+z_m^2+2r_r\cos\phi\left(d_{m,n}\cos\theta-z_m\sin\theta\right)\nonumber\\
\hspace{5.2cm}+2r_tr_r\left(\sin\phi\cos t-\cos\phi\sin t\cos\theta\right)-2d_{m,n}r_t\sin t\Big]^{-\frac{1}{2}}\nonumber\\
\hspace{5.6cm}\left(\sin t\cos \phi-\cos t\sin\phi \cos\theta\right) d\phi dt;
\label{eq:M}\\
M^t_{n_1,n_2}=\frac{\mu_0K_t^2r_t^2}{4\pi}\hspace{-0.2cm}\int_0^{2\pi}\hspace{-0.2cm}\int_0^{2\pi}
\Big[2r_t^2+(d'_{n_1,n_2})^2+2r_t^2\left(\sin\phi\cos t-\cos\phi\sin t\right)\nonumber\\
\hspace{4.5cm}+2r_td'_{n_1,n_2}\left(\cos\phi-\sin t\right)\Big]^{-\frac{1}{2}}\hspace{-0.2cm}\left(\sin t\cos \phi-\cos t\sin\phi\right) d\phi dt,
\label{eq:Mt}
\end{numcases}
\label{eq:M_Mt}
\end{subequations}
\hrulefill
\end{figure*}
\subsection{OAM-NFC Generation Scheme}
For the OAM-NFC generation, we denote by $\boldsymbol{\mathrm x}\hspace{-0.1cm}=\hspace{-0.1cm}\left[x_1,x_2,\cdots,x_{N_t}\right]^{T}$ the input signals, where ``$\left[\cdot\right]^T$" denotes the transpose operation. Since OAM-carrying vortex beams can be generated by a circular antenna array with equidistant elements using IDFT in RF domain\cite{oam_low_freq_radio}, the source voltages across the transmit coils, represented by $\boldsymbol{\mathrm v}^t\hspace{-0.1cm}=\hspace{-0.1cm}\left[v^t_1,v^t_2,\cdots,v^t_{N_t}\right]^{T}$, is given as follows:
\begin{align}
\setcounter{equation}{0}
\boldsymbol{\mathrm v}^t=\boldsymbol{\mathrm W}\boldsymbol{\mathrm x},
\label{eq:vt_x}
\end{align}
where $\boldsymbol{\mathrm W}$ is the unit IDFT matrix with the $n_1$th row and the $n_2$th column element given by ${W}_{n_1,n_2}\hspace{-0.1cm}=\hspace{-0.1cm}\frac{1}{\sqrt N_t}{\rm exp}[j{2\pi (n_1-1)(n_2-1)}/{N_t}]$. Taking the mutual inductance additive white Gaussian noise (AWGN)\cite{MI_WhiteNoise,MI_multi_access,MI_Underwater_overveiw} into consideration associated with Eq.~\eqref{eq:vt_x}, the received induced voltages across the receive coils, denoted by $\boldsymbol{\mathrm v}^r\hspace{-0.1cm}=\hspace{-0.1cm}\left[v^r_1,v^r_2,\cdots,v^r_{N_r}\right]^{T}$, can be given as follows:
\begin{align}
\boldsymbol{\mathrm v}^r=\boldsymbol{\mathrm H}\boldsymbol{\mathrm W}\boldsymbol{\mathrm x}+\boldsymbol{\mathrm n},
\end{align}
where $\boldsymbol{\mathrm n}\hspace{-0.1cm}=\hspace{-0.1cm}\left[n_1,\cdots,n_2,\cdots,n_{N_t}\right]^{T}$ is the AWGN corresponding to each receive coil and $\boldsymbol{\mathrm H}$ denotes the multi-coil mutual inductance channel given by Theorem~\ref{the:channel_matrix}.
\begin{theorem}
The multi-coil mutual inductance channel matrix is
\begin{align}
\boldsymbol{\mathrm H}=\frac{-j\omega}{Z_t}\boldsymbol{\mathrm M}-\frac{\omega^2}{Z_t^{2}}\boldsymbol{\mathrm M}\boldsymbol{\mathrm M}^t,
\label{eq:channel_matrix}
\end{align}
where $j$ is the imaginary symbol, $\omega$ is the angular frequency, $Z_t$ denotes the impedance of each transmit coil, $\boldsymbol{\mathrm M}$ denotes the mutual inductance between each pair of transmit coil and receive coil, and $\boldsymbol{\mathrm M}^t$ denotes the mutual inductance between each two transmit coils. $Z_t$ can be further given by $Z_t=R+{1}/{j\omega C}+j\omega L_t$, where $L_t$ denotes the self-inductance of each transmit coil. Since $r_t$ is relatively small, $L_t$ can be given by $L_t={\mu_0K^2_t\pi r_t}/{2}$. $C$ and $R$ denote the capacitance and resistance of each transmit coils, respectively. $C$ and $R$ are given as $C = 1/(\sqrt{w}L_t)$ and $R = 2\pi r_tK_t R_0 /S$, where $R_0$ denotes the resistivity depending on the type of the coil line and $S$ is the cross-sectional area of the coil line. The $m$th row and the $n$th column element of $\boldsymbol{\mathrm M}$, denoted by $M_{m,n}$ with $1\hspace{-0.1cm}\le\hspace{-0.1cm} n\hspace{-0.1cm}\le\hspace{-0.1cm} N_t$ and $1\hspace{-0.1cm}\le\hspace{-0.1cm} m\hspace{-0.1cm}\le\hspace{-0.1cm} N_r$, as well as the $n_2$th row and the $n_1$th column element of $\boldsymbol{\mathrm M}^t$, denoted by $M^t_{n_1,n_2}$ with $1\hspace{-0.1cm}\le\hspace{-0.1cm} n_1\ne n_2\hspace{-0.1cm}\le\hspace{-0.1cm} N_t$, are given as Eq.~\eqref{eq:M_Mt},
where $\mu_0$ is the magnetic permeability of vacuum, $d_{m,n}$ represents the distance from the center of the $m$th receive coil to the axis of the $n$th transmit coil, $d'_{n_1,n_2}$ denotes the distance between the centers of the $n_1$th transmit coil and the $n_2$th transmit coil, $z_m$ denotes the $z$ coordinate for the $m$th receive coil, and $\theta$ denotes the angle between the axes of the two coils. The main diagonal elements of $\boldsymbol{\mathrm M}^t$ are set as $0$.
\label{the:channel_matrix}
\end{theorem}
\begin{proof}
See Appendix~\ref{pro:multi-coil channel}.
\end{proof}

\subsection{OAM-NFC Detection Scheme Without Channel Estimation}
Next, we propose the OAM-NFC signals detection scheme without channel estimation for aligned transceivers and $N_r\hspace{-0.1cm}=\hspace{-0.1cm} IN_t$, where $I$ is a positive integer. In the first step, we take the mean for every $I$ rows in $\boldsymbol{\mathrm H}$ and formulate a circulant matrix $\hat{\boldsymbol{\mathrm H}}$ so that the channel matrix can be decomposed into a diagonal matrix by DFT matrix. $\hat{\boldsymbol{\mathrm H}}$ is given as follows:
\small{
\begin{align}
\setcounter{equation}{4}
\hat{\boldsymbol{\mathrm{H}}}\hspace{-0.1cm}=\hspace{-0.06cm}&\frac{1}{I}\hspace{-0.1cm}\sum_{i=1}^I\hspace{-0.1cm}
\begin{bmatrix}
{  H}_{i,1} & {H}_{i,n_2} & \cdots & {  H}_{i,N_t}\\
{  H}_{i+I,1} & {H}_{i+I,n_2} & \cdots & {  H}_{i+I,N_t}\\
\vdots & \vdots & \ddots & \vdots \\
{  H}_{i+I(n_1-1),1} & {H}_{i+I(n_1-1),n_2} & \cdots & {  H}_{i+I(n_1-1),N_t}\\
\vdots & \vdots & \ddots & \vdots \\
{  H}_{i+I(N_t-1),1} & {H}_{i+I(N_t-1),n_2} & \cdots & {  H}_{i+I(N_t-1),N_t}
\end{bmatrix},
\label{eq:H}
\end{align}
}
where ${H}_{i+I(n_1-1),n_2}$ is the [$i+I(n_1-1)$]th row and the $n_2$th column element of $\boldsymbol{\mathrm {H}}$. We also take the mean for every $I$ elements in $\boldsymbol{\mathrm v}^r$ and $\hat{\boldsymbol{\mathrm n}}$. New vectors are denoted by $\hat{\boldsymbol{\mathrm v}}^r \hspace{-0.1 cm}=\hspace{-0.1 cm} \frac{1}{I}\sum_{i=1}^I\left[ v^r_i, v^r_{i+I}, \cdots, v^r_{i+I(N_t-1)} \right]^T$ and $\hat{\boldsymbol{\mathrm n}} \hspace{-0.1 cm}=\hspace{-0.1 cm} \frac{1}{I}\sum_{i=1}^I\left[ n_i, n_{i+I}, \cdots, n_{i+I(N_t-1)} \right]^T$, respectively.
Clearly, $\hat{\boldsymbol{\mathrm v}}^r=\hat{\boldsymbol{\mathrm H}}\boldsymbol{\mathrm W}\boldsymbol{\mathrm x}+\hat{\boldsymbol{\mathrm n}}$. For mutual inductance AWGN channel, $\boldsymbol{\mathrm n}$ is a vector with independent and identically distributed elements following $\mathcal{CN}(0,N_0)$, where $N_0$ is the noise power spectral density. Thus, $\hat{\boldsymbol{\mathrm n}}$ is a vector with independent and identically distributed elements following $\mathcal{CN}(0,N_0/I)$.

In the second step, we make $N_t$-point unit DFT on $\hat{\boldsymbol{\mathrm v}}^r$. The signals after DFT, denoted by $\boldsymbol{\mathrm y}$, is given as follows:
\begin{align}
\boldsymbol{\mathrm y}&=\boldsymbol{\mathrm W}^H\hat{\boldsymbol{\mathrm v}}^r
=\boldsymbol{\mathrm W}^H\hat{\boldsymbol{\mathrm H}}\boldsymbol{\mathrm W}\boldsymbol{x}+\boldsymbol{\mathrm W}^H\hat{\boldsymbol{\mathrm n}},
\label{eq:y}
\end{align}
where $\boldsymbol{\mathrm W}^H$ is the conjugate transpose of $\boldsymbol{\mathrm W}$. For simplifying the following description, let
\begin{align}
\boldsymbol{\mathrm H}_{\rm {OAM}}=\boldsymbol{\mathrm W}^{H}\hat{\boldsymbol{\mathrm H}}\boldsymbol{\mathrm W}
\end{align}
denote the OAM-NFC channel matrix with its diagonal, denoted by $\boldsymbol{\mathrm h}_{\rm {OAM}}$, given as follows:
\begin{align}
\boldsymbol{\mathrm h}_{\rm {OAM}}={\rm diag}\left(\boldsymbol{\mathrm H}_{\rm {OAM}}\right).
\end{align}

In the third step, we recover the input signals, denoted by $\hat{\boldsymbol{\mathrm x}}$, from $\boldsymbol{\mathrm y}$ as follows:
\begin{align}
\hat{\boldsymbol{\mathrm x}}=\left[{\rm{Diag}}(\boldsymbol{\mathrm h}_{\rm {OAM}})\right]^{\dag}\boldsymbol{\mathrm y},
\label{eq:xrecover}
\end{align}
where ``$[\cdot]^{\dag}$" is the pseudo-inverse operation and ${\rm{Diag}}(\boldsymbol{\mathrm h}_{\rm {OAM}})$ denotes the diagonal matrix created by using the elements of $\boldsymbol{\mathrm h}_{\rm {OAM}}$. Finally, the output signal of the OAM-NFC system, denoted by $\boldsymbol{\tilde{\mathrm x}}$, can be given as follows:
\begin{align}
\boldsymbol{\tilde{\mathrm x}}\hspace{-0.1 cm}=\hspace{-0.1 cm}\mathop{\arg\min}\limits_{\boldsymbol{\mathrm x}\in \mathbb{C}_{N_t}}\left\Arrowvert{\hat{\boldsymbol{\mathrm x}}-\boldsymbol{\mathrm x}}\right\Arrowvert^2\hspace{-0.1 cm}=\hspace{-0.1 cm}\mathop{\arg\min}\limits_{\boldsymbol{\mathrm x}\in \mathbb{C}_{N_t}}\left\Arrowvert{\left[{\rm{Diag}}(\boldsymbol{\mathrm h}_{\rm {OAM}})\right]^{\dag}\boldsymbol{\mathrm W}^H\hat{\boldsymbol{\mathrm v}}^r \hspace{-0.1 cm}-\hspace{-0.1 cm}\boldsymbol{\mathrm x}}\right\Arrowvert^2,
\label{eq:judgment}
\end{align}
where $\mathbb{C}_{N_t}$ is the signal constellation. In fact, the detection in Eqs.~\eqref{eq:xrecover} and \eqref{eq:judgment} does not provide the optimal detection if the alignment is not perfect. Thus, we will combine the OAM-NFC detection with a simple channel estimation method later in this section to deal with the imperfect alignment.

\subsection{Feasibility of OAM-NFC Detection Schemes Without Channel Estimation}
To validate that input signals can be recovered by Eq.~\eqref{eq:xrecover}, we first prove that $\boldsymbol{\mathrm H}$ is a block circulant matrix for aligned transceivers with the following Theorem~\ref{the:circulant_matrix}.
\begin{theorem}
For aligned transmit and receive coil rings, $\boldsymbol{\mathrm H}$ is a block circulant matrix, given as follows:
\begin{align}
{\bf{H}} = \left[ {\begin{array}{*{20}{c}}
{{{\bf{h}}_1}}&{{{\bf{h}}_2}}& \ldots &{{{\bf{h}}_{{N_t}}}}\\
{{{\bf{h}}_{{N_t}}}}&{{{\bf{h}}_1}}& \ldots &{{{\bf{h}}_{{N_t} - 1}}}\\
 \vdots & \vdots & \ddots & \vdots \\
{{{\bf{h}}_2}}&{{{\bf{h}}_3}}& \ldots &{{{\bf{h}}_1}}
\end{array}} \right],
\end{align}
where $\boldsymbol{\mathrm h}_i\hspace{-0.1cm}=\hspace{-0.1cm}[H_{1,i},H_{2,i},\cdot,H_{I,i}]^T$, $1\hspace{-0.1cm}\le\hspace{-0.1cm}i\hspace{-0.1cm}\le\hspace{-0.1cm}N_t$.
\label{the:circulant_matrix}
\end{theorem}
\begin{proof}
See Appendix~\ref{pro:circulant channel}.
\end{proof}
Based on Theorem~\ref{the:circulant_matrix}, Eq.~\eqref{eq:H} can be rewritten as follows:
\begin{align}
\hat{\boldsymbol{\mathrm{H}}}=&\frac{1}{I}\hspace{-0.1cm}\sum_{i=1}^I\hspace{-0.1cm}
\begin{bmatrix}
{  H}_{i,1} & {  H}_{i,2} & \cdots & {  H}_{i,N_t}\\
{  H}_{i,N_t} & {  H}_{i,1} & \cdots & {  H}_{i,N_t-1}\\
\vdots & \vdots & \ddots & \vdots \\
{  H}_{i,2} & {  H}_{i,3} & \cdots & {  H}_{i,1}
\end{bmatrix}.
\end{align}
Thus, $\hat{\boldsymbol{\mathrm H}}$ is a circulant matrix.

In view of that the cyclic shift relationships in a circulant matrix and the circular convolution, given as follows\cite{oam_zc_scra}:
\begin{align}
\boldsymbol{\mathrm W}^{H}(\boldsymbol{\mathrm C} \boldsymbol{\mathrm v}) = \sqrt{N_t}(\boldsymbol{\mathrm W}^{H}\boldsymbol{\mathrm c}) \cdot (\boldsymbol{\mathrm W}^{H} \boldsymbol{\mathrm v}),
\label{eq:idenequ}
\end{align}
Eq.~\eqref{eq:y} can be simplified as follows:
\begin{align}
\boldsymbol{\mathrm y}&=\boldsymbol{\mathrm W}^H\hat{\boldsymbol{\mathrm H}}\left(\boldsymbol{\mathrm W}\boldsymbol{x}\right)+\boldsymbol{\mathrm W}^H\hat{\boldsymbol{\mathrm n}}\nonumber
\\&=\sqrt{N_t}\left[\boldsymbol{\mathrm W}^H\hat{\boldsymbol{\mathrm H}}(:,1)\right] \cdot \left[\boldsymbol{\mathrm W}^H\boldsymbol{\mathrm W}\boldsymbol{x}\right]+\boldsymbol{\mathrm W}^H\hat{\boldsymbol{\mathrm n}}\nonumber
\\&=\sqrt{N_t}\left[\boldsymbol{\mathrm W}^H\hat{\boldsymbol{\mathrm H}}(:,1)\right] \cdot \boldsymbol{x} + \boldsymbol{\mathrm W}^H\hat{\boldsymbol{\mathrm n}},
\label{eq:y2}
\end{align}
where $\boldsymbol{\mathrm C}$ represents a circulant matrix, $\boldsymbol{\mathrm c}$ is the first column of $\boldsymbol{C}$, $\boldsymbol{\mathrm v}$ denotes a general column vector, ``$\cdot$" represents dot product of matrices, and $\hat{\boldsymbol{\mathrm H}}(:,1)$ is the first column of $\hat{\boldsymbol{\mathrm H}}$.
In the third step, we let $\boldsymbol{\mathrm h}^{\rm Align}_{\rm {OAM}}\hspace{-0.1cm}=\hspace{-0.1cm}\sqrt{N_t}\left[\boldsymbol{\mathrm W}^H\hat{\boldsymbol{\mathrm H}}(:,1)\right]$. Based on Eq.~\eqref{eq:y2}, the relationships between $\boldsymbol{\mathrm H}^{\rm OAM}$ and $\boldsymbol{\mathrm h}^{\rm Align}_{\rm {OAM}}$ can be given as follows:
\begin{align}
\boldsymbol{\mathrm H}^{\rm OAM} = {\rm Diag}\left(\boldsymbol{\mathrm h}^{\rm Align}_{\rm {OAM}}\right).
\end{align}
Therefore, for aligned transceivers with $N_r\hspace{-0.1cm}=\hspace{-0.1cm} IN_t$, $\boldsymbol{\mathrm H}^{\rm OAM}$ is a diagonal matrix and is equal to ${\rm Diag}\left(\boldsymbol{\mathrm h}_{\rm {OAM}}\right)$. Hence, for aligned transceivers, Eq.~\eqref{eq:y2} can be rewritten as follows:
\begin{align}
\boldsymbol{\mathrm y}&={\rm Diag}\left(\boldsymbol{\mathrm h}_{\rm {OAM}}\right)\boldsymbol{x} + \boldsymbol{\mathrm W}^H\hat{\boldsymbol{\mathrm n}}.
\label{eq:y3}
\end{align}
Finally, based on Eqs.~\eqref{eq:judgment} and \eqref{eq:y3}, we can obtain
\begin{align}
\hat{\boldsymbol{\mathrm x}}=\boldsymbol{\mathrm x}+\left[{\rm{Diag}}(\boldsymbol{\mathrm h}_{\rm {OAM}})\right]^{\dag}\boldsymbol{\mathrm W}^H\hat{\boldsymbol{\mathrm n}},
\label{eq:xrecover2}
\end{align}
which validates the feasibility of our proposed OAM-NFC detection scheme without channel estimation. Eq.~\eqref{eq:xrecover2} also indicates that for aligned transceivers, crosstalks caused by the inductive coupling among different transmit coils are not obstacles to the recovery of OAM carrying signals with our proposed OAM-NFC generation and detection schemes.

\subsection{OAM-NFC Detection Scheme With Channel Estimation}
For misaligned transceivers or $N_r\hspace{-0.1cm}\ne\hspace{-0.1cm} IN_t$, we give a detection scheme with channel estimation. We send a pilot, denoted by an $N_t\times T$ matrix $\boldsymbol{\tilde{\mathrm S}}$, which is created by applying cyclic shifts to Zadoff-Chu (ZC) root sequences with length $T$ and its $n_t$th row and $t$th column element, denoted by ${\rm \tilde{S}}_{n_t, t}$, is given as follows:
\begin{align}
{\rm \tilde{S}}_{n_t, t}=
\begin{cases}
e^{-j\frac{\pi}{T}p\left[(t-n_t)^2\right]_T},&\hspace{-0.2cm}T\ \rm{is even};\\
e^{-j\frac{\pi}{T}p\left[(t-n_t)(t-n_t+1)\right]_T},&\hspace{-0.2cm}T\ \rm{is\ odd},
\end{cases}
\end{align}
where ``$[\cdot]_T$'' denotes the modulo operation of $T$ and $p$ is a positive integer relatively prime with $T$. $T$ should be larger than both $N_t$ and $N_r$. The received pilot signal, denoted by $\boldsymbol{\tilde{\mathrm R}}$, can be given as follows:
\begin{align}
\boldsymbol{\tilde{\mathrm R}}=\sqrt{P}\boldsymbol{\mathrm H}\boldsymbol{\tilde{\mathrm S}}+\boldsymbol{\mathrm N},
\end{align}
where $P$ denotes the transmit signal to noise ratio (SNR) for each transmit coil and $\boldsymbol{\mathrm N}$ denotes the AWGN matrix with independent and identically distributed elements following $\mathcal{CN}(0,1)$.
Adopting the least square (LS) method, the estimated channel matrix is given as follows\cite{MIMO_channel_estimation}:
\begin{align}
\boldsymbol{\mathrm H}^e=\frac{1}{\sqrt{P}T}\boldsymbol{\tilde{\mathrm R}}\boldsymbol{\tilde{\mathrm S}}^{\rm H}
=\boldsymbol{\mathrm H}+\frac{1}{\sqrt{P}T}\boldsymbol{\mathrm N}\boldsymbol{\tilde{\mathrm S}}^{\rm H}.
\end{align}
As the value of $P$ increases to infinity, $\boldsymbol{\mathrm H}^e$ converges $\boldsymbol{\mathrm H}$.
Then, the LS recovered signals, denoted by $\hat{\boldsymbol{\mathrm x}}^{\rm LS}$, can be given as follows:
\begin{align}
\hat{\boldsymbol{\mathrm x}}^{\rm LS} = \boldsymbol{\mathrm W}^{H}[\boldsymbol{\mathrm H}^e]^{\dag}\boldsymbol{\mathrm v}^r.
\label{eq:xrecovered_LS}
\end{align}
Finally, the output signal of the detection with channel estimation, denoted by $\boldsymbol{\tilde{\mathrm x}}^{\rm LS}$, can be given as follows:
\begin{align}
\boldsymbol{\tilde{\mathrm x}}^{\rm LS}\hspace{-0.1cm}=\mathop{\arg\min}\limits_{\boldsymbol{\mathrm x}\in \mathbb{C}_{N_t}}\left\Arrowvert{\hat{\boldsymbol{\mathrm x}}^{\rm LS}-\boldsymbol{\mathrm x}}\right\Arrowvert^2
\hspace{-0.1cm}=\mathop{\arg\min}\limits_{\boldsymbol{\mathrm x}\in \mathbb{C}_{N_t}}\left\Arrowvert{\boldsymbol{\mathrm W}^{H}[\boldsymbol{\mathrm H}^e]^{\dag}\boldsymbol{\mathrm v}^r -\boldsymbol{\mathrm x}}\right\Arrowvert^2.
\label{eq:judgment_LS}
\end{align}

\subsection{Feasibility of OAM-NFC Detection Schemes With Channel Estimation}
First, Eq.~\eqref{eq:xrecovered_LS} can be simplified as follows:
\begin{align}
\hat{\boldsymbol{\mathrm x}}^{\rm LS}
 &= \boldsymbol{\mathrm W}^{H}[\boldsymbol{\mathrm H}^e]^{\dag}\boldsymbol{\mathrm v}^r
 = \boldsymbol{\mathrm W}^{H}[\boldsymbol{\mathrm H}^e]^{\dag}\boldsymbol{\mathrm H}\boldsymbol{\mathrm W}\boldsymbol{\mathrm x} + \boldsymbol{\mathrm W}^{H}[\boldsymbol{\mathrm H}^e]^{\dag}\boldsymbol{\mathrm n}.
\end{align}
Then, the mean square error (MSE) per symbol for $\hat{\boldsymbol{\mathrm x}}^{\rm LS}$, denoted by $\boldsymbol{\mathrm\eta}(\hat{\boldsymbol{\mathrm x}}^{\rm LS})$, can be given as follows:
\begin{align}
\boldsymbol{\mathrm\eta}(\hat{\boldsymbol{\mathrm x}}^{\rm LS})&\hspace{-0.1cm}=\hspace{-0.1cm}\frac{1}{N_t}\mathbb E \left\{\Arrowvert\hat{\boldsymbol{\mathrm x}}^{\rm LS}-\boldsymbol{\mathrm x}\Arrowvert^2\right\}\nonumber
\\&\hspace{-0.1cm}=\hspace{-0.1cm} \frac{1}{N_t}\mathbb E \left\{\hspace{-0.05cm} \Arrowvert\boldsymbol{\mathrm W}^{H}[\boldsymbol{\mathrm H}^e]^{\dag}\boldsymbol{\mathrm H}\boldsymbol{\mathrm W}\boldsymbol{\mathrm x}\hspace{-0.1cm}-\hspace{-0.1cm}\boldsymbol{\mathrm x}\Arrowvert^2 \hspace{-0.05cm}\right\} \hspace{-0.1cm}+\hspace{-0.1cm} \frac{\mathbb E \left\{\hspace{-0.05cm} \left|{\rm diag}\left([\boldsymbol{\mathrm H}^e]^{\dag}\right)\right|^2 \hspace{-0.05cm}\right\}}{N_t}N_0,
\label{eq:MSE_xMMES}
\end{align}
where $\mathbb E\{\cdot\}$ denotes the average operation. As SNR is increased to infinity, Eq.~\eqref{eq:MSE_xMMES} becomes
\begin{align}
\lim\limits_{P\to\infty}\boldsymbol{\mathrm\eta}(\hat{\boldsymbol{\mathrm x}}^{\rm LS})&=\frac{\left|{\rm diag}\left(\boldsymbol{\mathrm H}^{\dag}\right)\right|^2}{N_t}N_0,
\end{align}
which validates the feasibility of our proposed OAM-NFC detection scheme with channel estimation for unaligned transceivers given in Eq.~\eqref{eq:judgment_LS}.

\section{Channel Capacity Analyses}\label{sec:Capacity}
In this section, we first derive the channel capacities of our proposed OAM-NFC with and without channel estimation. Then, we give the numerical analyses of OAM-NFC capacities with and without channel estimation under imperfect alignment condition. After that, we analyze how $N_t$, $N_r$, $D$, $R_t$, $R_r$, $r_t$, $r_r$,  $\theta_x$, $\theta_y$, $d_x$ and $d_y$ impact the OAM-NFC capacity. We also give the channel capacities of SISO-NFC and MIMO-NFC for comparison.

\subsection{OAM-NFC Channel Capacity and BER Without Channel Estimation}
\begin{figure*}[tp]
\setcounter{equation}{29}
\begin{subequations}
\begin{numcases}{}
\widehat{C}^{\rm {upper}}_{\rm {OAM}}=N_t{\rm{log}}\hspace{-0.1cm}\left(1\hspace{-0.1cm}+\hspace{-0.1cm}\frac{P_tI^2}{N_0}\frac{4\mu^2_0\omega^2K_t^2K_r^2r_tr_r}{\pi^2|Z_t|^2}\left|\int_0^\pi \frac{\left(1-|R_r\hspace{-0.1cm}-\hspace{-0.1cm}R_t|/r_t\right)\cos\phi\Psi(\widehat k^{upper})}{\widehat k^{upper}\sqrt{(\widehat V^{upper})^3}}d\phi\right|^2\right);\nonumber\label{eq:COAM2_upper}\\ \\
\widehat{C}^{\rm {lower}}_{\rm {OAM}}=N_t{\rm{log}}\hspace{-0.1cm}\left(1\hspace{-0.1cm}+\hspace{-0.1cm}\frac{P_tI^2}{N_0}\frac{4\mu^2_0\omega^2K_t^2K_r^2r_tr_r}{\pi^2|Z_t|^2}\left|\int_0^\pi \frac{\left[1-\left(R_r+R_t\right)/r_t\right]\cos\phi\Psi(\widehat k^{lower})}{\widehat k^{lower}\sqrt{(\widehat V^{lower})^3}}d\phi\right|^2\right),\nonumber\label{eq:COAM2_lower}\\
\end{numcases}\label{eq:COAM2_upper_lower}
\end{subequations}
\hrulefill
\end{figure*}
We first analyze the OAM-NFC channel capacity and the corresponding BER without channel estimation. In this section, we let $\boldsymbol{\mathrm x}$ be a complex Gaussian distributed vector with zero mean and the total transmit power per symbol period, denoted by $P_t$, is given by $P_t \hspace{-0.1cm}=\hspace{-0.1cm} \mathbb E \left\{\boldsymbol{\mathrm x}^H\boldsymbol{\mathrm x}\right\}$. Since $\hat{\boldsymbol{\mathrm n}}$ is a vector with independent and identically distributed elements following $\mathcal{CN}(0,N_0/I)$, based on Eq.~\eqref{eq:y}, when the transmit power is equally allocated to each OAM mode, the channel capacity of OAM-NFC, denoted by ${C}_{\rm {OAM}}$, can be derived as follows:
\begin{align}
\setcounter{equation}{25}
{C}_{\rm {OAM}}
\hspace{-0.1cm}=\hspace{-0.1cm}\sum_{l=0}^{N_t-1}{\rm{log}}\hspace{-0.1cm}\left(1 \hspace{-0.1cm}+\hspace{-0.1cm}\frac{\left|\boldsymbol{\mathrm h}_{\rm {OAM}}(l\hspace{-0.1cm}+\hspace{-0.1cm}1)\right|^2}{\frac{N_0N_t}{P_tI}\hspace{-0.1cm}+\hspace{-0.1cm} \sum_{q=0,q\ne l}^{Nt-1}\left|\boldsymbol{\mathrm H}_{\rm {OAM}}(q\hspace{-0.1cm}+\hspace{-0.1cm}1,l\hspace{-0.1cm}+\hspace{-0.1cm}1)\right|^2}\right),
\label{eq:COAM}
\end{align}
where $\boldsymbol{\mathrm h}_{\rm {OAM}}(l+1)$ represents the $(l+1)$th element of $\boldsymbol{\mathrm h}_{\rm {OAM}}$ and $\boldsymbol{\mathrm H}_{\rm {OAM}}(q+1,l+1)$ denotes the $(q+1)$th row and the $(l+1)$th column element in $\boldsymbol{\mathrm H}_{\rm {OAM}}$. Since Eq.~\eqref{eq:COAM} is the sum for multiple log functions of signal to interference plus noise ratio (SINR), the OAM-NFC capacity can be further increased by using existing analog or digital interference cancellation techniques\cite{OAM_analog_interfree,OAM_SIC}. Basically, the purpose of this paper is to provide a kind of benchmark that verifies the feasibility and capacity enhancement of OAM-NFC. Studies for algorithms that can further increase the OAM-NFC capacity are open problems. Using binary phase shift keying (BPSK) for the input and output signals, the bit error rate (BER) of OAM-NFC, denoted by ${P_e}_{\rm {OAM}}$, can be given as follows:
\begin{align}
{P_e}_{\rm {OAM}}\hspace{-0.1cm}=\hspace{-0.1cm}\frac{1}{2N_t}\sum_{l=0}^{N_t-1}{\rm{erfc}}\sqrt{\frac{\left|\boldsymbol{\mathrm h}_{\rm {OAM}}(l\hspace{-0.1cm}+\hspace{-0.1cm}1)\right|^2}{\frac{N_0N_t}{P_tI}\hspace{-0.1cm}+\hspace{-0.1cm}\boldsymbol \sum_{q=0,q\ne l}^{Nt-1}\left|\boldsymbol{\mathrm H}_{\rm {OAM}}(q\hspace{-0.1cm}+\hspace{-0.1cm}1,l\hspace{-0.1cm}+\hspace{-0.1cm}1)\right|^2}},
\label{eq:BEROAM}
\end{align}
where erfc$(\alpha)=\frac{2}{\sqrt{\pi}}$$\int_{\alpha}^{\infty} e^{-t^2}dt$.

In Eq.~\eqref{eq:channel_matrix}, because the values of elements in $\frac{\omega^2}{Z_t^{2}}\boldsymbol{\mathrm M}\boldsymbol{\mathrm M}^t$ are much smaller than those in $\boldsymbol{\mathrm H}$ (for example, for $K_t\hspace{-0.1cm}=\hspace{-0.1cm}K_r\hspace{-0.1cm}=\hspace{-0.1cm}5$, $N_t\hspace{-0.1cm}=\hspace{-0.1cm}N_r\hspace{-0.1cm}=\hspace{-0.1cm}8$, $R_t\hspace{-0.1cm}=\hspace{-0.1cm}R_r\hspace{-0.1cm}=\hspace{-0.1cm}25$ mm, $r_t\hspace{-0.1cm}=\hspace{-0.1cm}r_r\hspace{-0.1cm}=\hspace{-0.1cm}5$ mm, $D\hspace{-0.1cm}=\hspace{-0.1cm}25$ mm, $\theta_x\hspace{-0.1cm}=\hspace{-0.1cm}\theta_y\hspace{-0.1cm}=\hspace{-0.1cm}10$ deg, and $d_x\hspace{-0.1cm}=\hspace{-0.1cm}d_y\hspace{-0.1cm}=\hspace{-0.1cm}10$ mm, the values of elements in $\frac{\omega^2}{Z_t^{2}}\boldsymbol{\mathrm M}\boldsymbol{\mathrm M}^t$ are $8$ orders of magnitude smaller than those in $\boldsymbol{\mathrm H}$), we ignore the mutual inductance between each two transmit coils. Noting that for aligned transmit and receive coil rings, $\boldsymbol{\mathrm M}$ is a block circulant matrix with $\theta\hspace{-0.1cm}=\hspace{-0.1cm}0$ and $d_x\hspace{-0.1cm}=\hspace{-0.1cm}d_y\hspace{-0.1cm}=\hspace{-0.1cm}0$. Thus, the simplified OAM-NFC capacity, denoted by $\widehat{C}_{\rm {OAM}}$, can be derived as follows:
\begin{align}
\widehat{C}_{\rm {OAM}}&\hspace{-0.1cm}=\hspace{-0.1cm}\sum_{l=0}^{N_t-1}{\rm{log}}\hspace{-0.1cm}\left(1\hspace{-0.1cm}+\hspace{-0.1cm}\frac{P_tI}{N_0N_t}\left|\boldsymbol{\mathrm h}_{\rm {OAM}}(l+1)\right|^2\right)\nonumber\\
&\hspace{-0.1cm}=\hspace{-0.1cm}\sum_{l=0}^{N_t-1}{\rm{log}}\hspace{-0.1cm}\left(1\hspace{-0.1cm}+\hspace{-0.1cm}\frac{P_tI}{N_0N_t}\frac{\omega^2}{|Z_t|^2N_r}\left|\sum_{m=1}^{N_r}e^{-jl\frac{2\pi m}{N_r}}M_{m,1}\right|^2\right)\nonumber\\
&\hspace{-0.1cm}=\hspace{-0.1cm}\sum_{l=0}^{N_t-1}{\rm{log}}\Bigg(1\hspace{-0.1cm}+\hspace{-0.1cm}\frac{P_t}{N_0N_t^2}\frac{4\mu^2_0\omega^2K_t^2K_r^2r_tr_r}{\pi^2|Z_t|^2}\nonumber\\
&\hspace{1cm}\left|\sum_{m=1}^{IN_t}e^{-jl\frac{2\pi m}{N_r}}\hspace{-0.1cm}\int_0^\pi \frac{\left(1- d_{m,1}/r_t\right)\cos\phi\Psi(\widehat k_{m,1})}{\widehat k_{m,1}\sqrt{\widehat V_{m,1}^3}}d\phi\right|^2\Bigg),
\label{eq:COAM2}
\end{align}
where
\begin{align}
\left\{\!\!
\begin{array}{ll}
d_{m,1}=\sqrt{R_r^2+R_t^2-2R_rR_t\cos\left[2\pi\left(\frac{m}{N_r}-\frac{n}{N_t}\right)\right]};\\
\\
\widehat{V}_{m,1}=\sqrt{1+\left({ d_{m,1}}/{r_r}\right)^2-2\left({ d_{m,1}}/{r_r}\right)\cos\phi};\\
\\
\widehat{k}_{m,1}=2\sqrt{\frac{r_r\widehat{V}_{m,1}}{\left(r_t+r_r\widehat{V}_{m,1}\right)^2+D^2}};\\
\\
\Psi(\widehat k_{m,1})=\left(1-\frac{\widehat k_{m,1}^2}{2}\right)K(\widehat k_{m,1})-E(\widehat k_{m,1}).
\end{array}
\right.
\label{eq:VKPsi2}
\end{align}
In Eq.~\eqref{eq:VKPsi2}, $K(\xi)$ and $E(\xi)$ are the first and the second kinds of elliptic integrals. The upper and lower limits of Eq.~\eqref{eq:COAM2} are given by Theorem~\ref{the:upper_lower} as follows.
\begin{theorem}
The upper and lower limits of $\widehat{C}_{\rm {OAM}}$, denoted by $\widehat{C}^{\rm {upper}}_{\rm {OAM}}$ and $\widehat{C}^{\rm {lower}}_{\rm {OAM}}$, respectively, are given as Eq.~\eqref{eq:COAM2_upper_lower},
where $\widehat V^{upper}\hspace{-0.1cm}=\hspace{-0.1cm}\widehat V_{m,1}\Big|_{d_{m,1}=|R_r-R_t|}$, $\widehat k^{upper}\hspace{-0.1cm}=\hspace{-0.1cm}\widehat k_{m,1}\Big|_{\widehat V_{m,1}=\widehat V^{upper}}$, $\widehat V^{lower}\hspace{-0.1cm}=\hspace{-0.1cm}\widehat V_{m,1}\Big|_{d_{m,1}=R_r+R_t}$, and $\widehat k^{lower}\hspace{-0.1cm}=\hspace{-0.1cm}\widehat k_{m,1}\Big|_{\widehat V_{m,1}=\widehat V^{lower}}$.
\label{the:upper_lower}
\end{theorem}
\begin{proof}
See Appendix~\ref{pro:COAM2_upper_lower}.
\end{proof}

\subsection{OAM-NFC Channel Capacity and BER With Channel Estimation}
\begin{figure*}[tp]
\setcounter{equation}{30}
\begin{subequations}
\begin{numcases}{}
C^{\rm LS}_{\rm {OAM}}=\sum_{l=0}^{N_t-1}{\rm{log}}\left(1+\frac{\frac{P_t}{N_t}\left|\boldsymbol{\mathrm H}^{\rm LS}_{\rm OAM}(l+1,l+1)\right|^2}{\frac{Pt}{Nt}\sum_{q=0,j\ne l}^{N_t-1}\left|\boldsymbol{\mathrm H}^{\rm LS}_{\rm OAM}(l+1,q+1)\right|^2+N_0 \boldsymbol{\mathrm D}^{\rm LS}_{\rm OAM}(l+1,l+1)}\right);\nonumber\\ \label{eq:C^LS_OAM}\\
P^{\rm LS}_{\rm {eOAM}}=\frac{1}{2N_t}\sum_{l=0}^{N_t-1}{\rm{erfc}}\sqrt{\frac{\frac{P_t}{N_t}\left|\boldsymbol{\mathrm H}^{\rm LS}_{\rm OAM}(l+1,l+1)\right|^2}{\frac{Pt}{Nt}\sum_{q=0,j\ne l}^{N_t-1}\left|\boldsymbol{\mathrm H}^{\rm LS}_{\rm OAM}(l+1,q+1)\right|^2+N_0 \boldsymbol{\mathrm D}^{\rm LS}_{\rm OAM}(l+1,l+1)}},\nonumber\\ \label{eq:Pe^LS_OAM}
\end{numcases}\label{eq:C^LS_Pe^LS_OAM}
\end{subequations}
\hrulefill
\end{figure*}
Then, we derive the OAM-NFC channel capacity and the corresponding BER with channel estimation. Based on Eq.~\eqref{eq:judgment_LS}, the OAM-NFC channel capacity and the corresponding BER with channel estimation, denoted by $C^{\rm LS}_{\rm {OAM}}$ and $P^{\rm LS}_{\rm {eOAM}}$, respectively, can be given as Eq.~\eqref{eq:C^LS_Pe^LS_OAM},
where $\boldsymbol{\mathrm H}^{\rm LS}_{\rm OAM}\hspace{-0.1cm}=\hspace{-0.1cm}\boldsymbol{\mathrm W}^H(\boldsymbol{\mathrm H}^e)^{\dag}\boldsymbol{\mathrm H}\boldsymbol{\mathrm W}$ and $\boldsymbol{\mathrm D}^{\rm LS}_{\rm OAM}\hspace{-0.1cm}=\hspace{-0.1cm}\boldsymbol{\mathrm W}^H(\boldsymbol{\mathrm H}^e)^\dag\left[\boldsymbol{\mathrm W}^H(\boldsymbol{\mathrm H}^e)^\dag\right]^H$ denote the OAM-NFC channel matrix and OAM-NFC equalization matrix with LS, respectively.

\begin{figure*}[htbp]
\centering
\subfigure[OAM-NFC capacity without channel estimation.]{
\begin{minipage}{0.45\linewidth}
\centering
\includegraphics[scale=0.6]{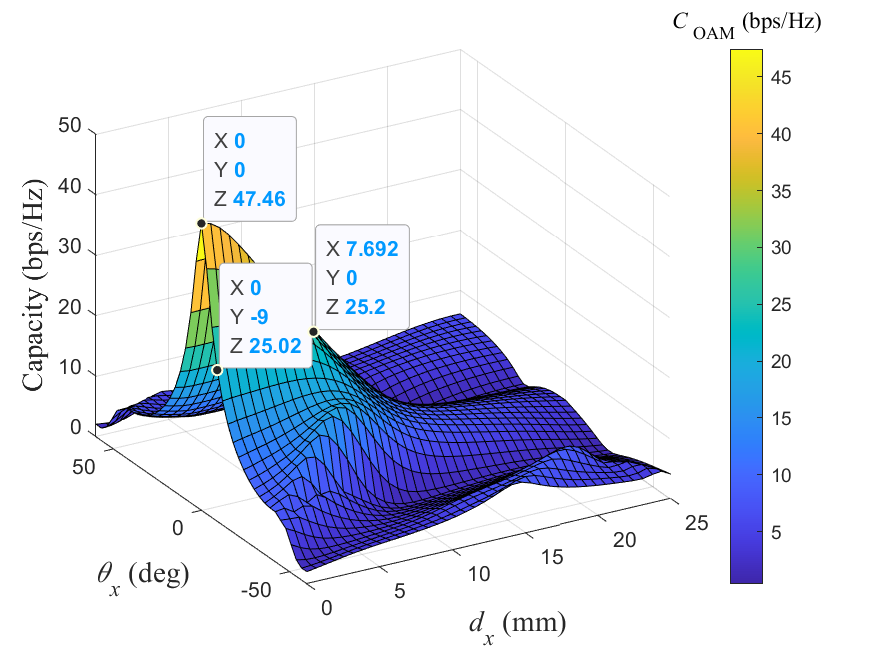}
\label{fig:Capacity_theo_W_20dB}
\end{minipage}
}
\subfigure[OAM-NFC capacity with channel estimation.]{
\begin{minipage}{0.45\linewidth}
\centering
\includegraphics[scale=0.6]{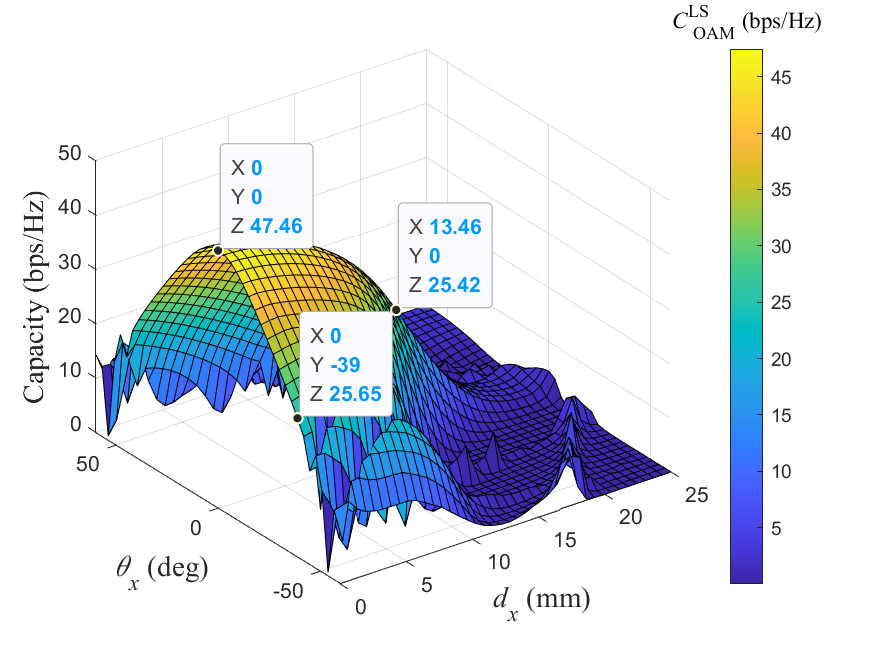}
\label{fig:Capacity_theo_LS_20dB}
\end{minipage}
}
\centering
\caption{OAM-NFC capacities with and without channel estimation under imperfect alignment condition.}
\label{fig:Capacity_theo_LSvsW}
\end{figure*}

\subsection{Analyses of OAM-NFC Capacity Under Imperfect Alignment Condition}\label{sec:capacity_analyses_LSvsW}
In this part, we give the numerical analyses of OAM-NFC capacities with and without channel estimation based on Eqs.~\eqref{eq:COAM} and \eqref{eq:C^LS_OAM} under imperfect alignment conditions. In Fig.~\ref{fig:Capacity_theo_LSvsW}, we set the frequency at $13.56$ MHz while let $N_t\hspace{-0.1cm}=\hspace{-0.1cm}N_r\hspace{-0.1cm}=\hspace{-0.1cm}8$, $R_t\hspace{-0.1cm}=\hspace{-0.1cm}R_r\hspace{-0.1cm}=\hspace{-0.1cm}25$ mm, $D\hspace{-0.1cm}=\hspace{-0.1cm}25$ mm, $d_y\hspace{-0.1cm}=\hspace{-0.1cm}\theta_y\hspace{-0.1cm}=\hspace{-0.1cm}0$, $r_t\hspace{-0.1cm}=\hspace{-0.1cm}r_r\hspace{-0.1cm}=\hspace{-0.1cm}5$ mm, $K_t\hspace{-0.1cm}=\hspace{-0.1cm}K_r\hspace{-0.1cm}=\hspace{-0.1cm}1$, $R_0\hspace{-0.1cm}=\hspace{-0.1cm}0.0175$ $\rm\Omega/\rm m/\rm mm^2$, and $S\hspace{-0.1cm}=\hspace{-0.1cm}0.05\ \rm mm^2$. $P_t$ is set as $8$ W and $N_0$ is set as $0.08$ W. We let $d_x$ vary from $0$ to $25$ mm as well as $\theta_x$ from $-60$ to $60$ degrees to analyze the feasibility of our proposed OAM-NFC under imperfect alignment conditions. Fig.~\ref{fig:Capacity_theo_W_20dB} shows the OAM-NFC capacity without channel estimation corresponding to different values of $d_x$ and $\theta_x$. As $|\theta_x|$ increases to $9$ degrees or $d_x$ increases to $7.5$ mm, the value of $C_{\rm {OAM}}$ decreases by $3$ dB compared with that of perfect alignment condition. Fig.~\ref{fig:Capacity_theo_LS_20dB} shows the OAM-NFC capacity with channel estimation corresponding to different values of $d_x$ and $\theta_x$. As $|\theta_x|$ increases to $40$ degrees or $d_x$ increases to $13.5$ mm, the value of $C^{\rm LS}_{\rm {OAM}}$ decreases by $3$ dB. Compare Figs.~\ref{fig:Capacity_theo_W_20dB} and \ref{fig:Capacity_theo_LS_20dB}, it can be found that although $C_{\rm {OAM}}$ is equal to $C^{\rm LS}_{\rm {OAM}}$ under perfect alignment conditinon, $C^{\rm LS}_{\rm {OAM}}$ decreases much slower than $C_{\rm {OAM}}$ as $d_x$ and $|\theta_x|$ increase, which indicates that channel estimation can improve the robustness of our proposed OAM-NFC against misalignment.

\subsection{Analyses of OAM-NFC Channel Capacity}\label{sec:capacity_analyses}
In this part, we analyze how various parameters affect our proposed OAM-NFC channel capacity. Eqs.~\eqref{eq:COAM} and \eqref{eq:COAM2} indicate that the obtained OAM-NFC capacity increases as $\omega$, $K_t$, and $K_r$ increase while decreases as $|Z_t|$ increases. Thus, to achieve high capacity, it is demanded to use coils with more turns and smaller impedance at a higher frequency. Then, based on Eqs.~\eqref{eq:COAM2}, \eqref{eq:COAM2_upper}, and \eqref{eq:COAM2_lower}, we analyze how $N_t$, $N_r$, $D$, $R_t$, $R_r$, $r_t$, and $r_r$ impact the OAM-NFC capacity. Also, based on Eq.~\eqref{eq:C^LS_OAM}, numerical results are given to validate our analyses.

\begin{figure*}[htbp]
\centering
\subfigure[Impact of transmit and receive coil numbers.]{
\begin{minipage}{0.45\linewidth}
\centering
\includegraphics[scale=0.55]{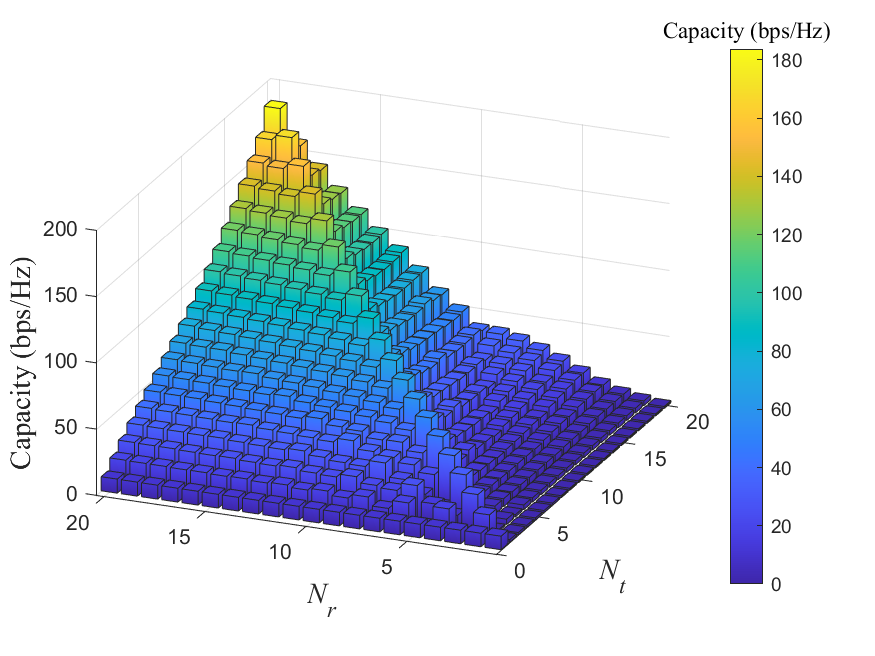}
\label{fig:Capacity_theo_N}
\end{minipage}
}
\subfigure[Impact of transceiver distance.]{
\begin{minipage}{0.45\linewidth}
\centering
\includegraphics[scale=0.52]{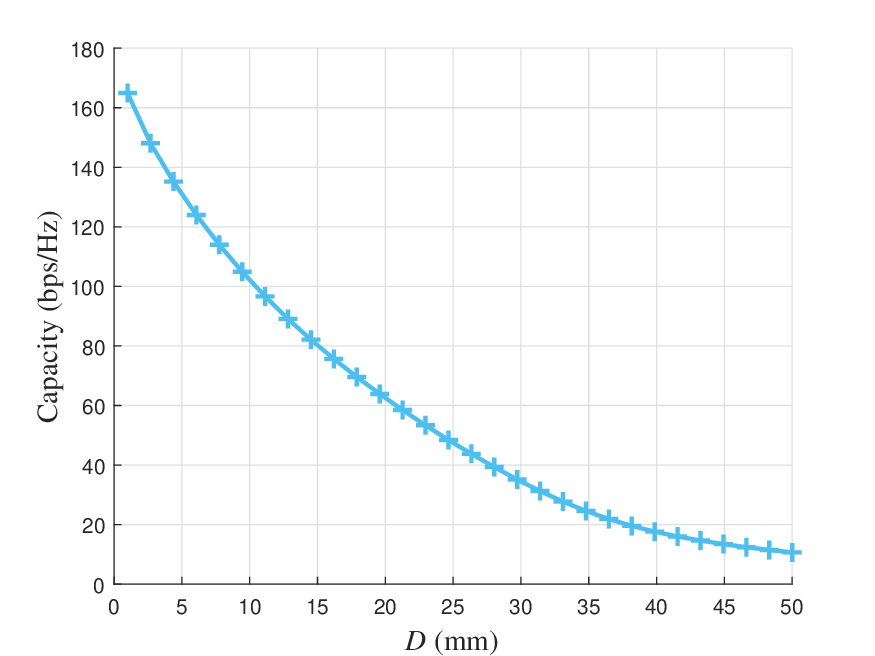}
\label{fig:Capacity_theo_D}
\end{minipage}
}\\
\subfigure[Impact of transmit and receive coil ring radii.]{
\begin{minipage}{0.45\linewidth}
\centering
\includegraphics[scale=0.55]{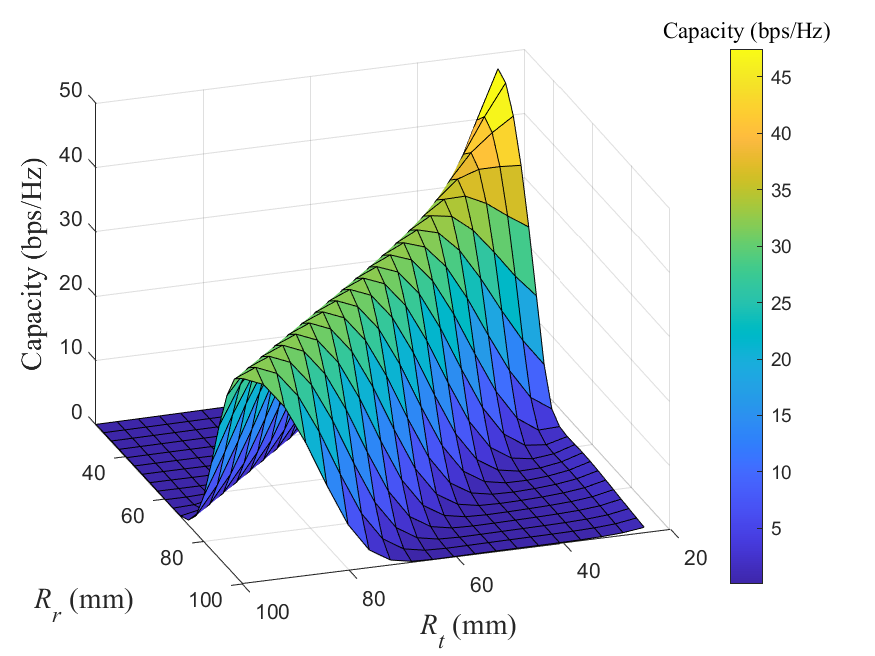}
\label{fig:Capacity_theo_R}
\end{minipage}
}
\subfigure[Impact of transmit and receive coil radii.]{
\begin{minipage}{0.45\linewidth}
\centering
\includegraphics[scale=0.55]{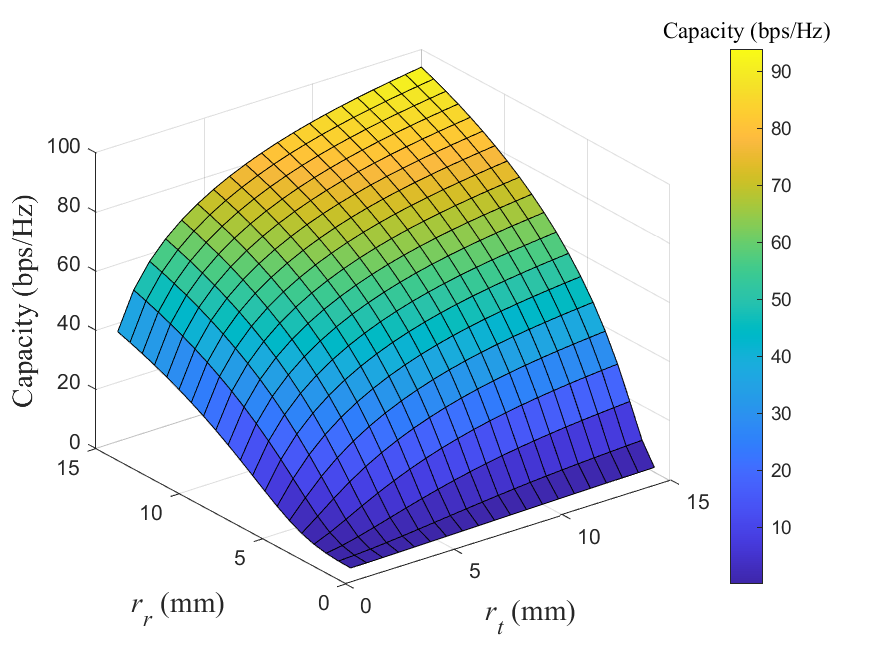}
\label{fig:Capacity_theo_rr}
\end{minipage}
}
\centering
\caption{Numerical analyses of the OAM-NFC capacity with channel estimation at the receiver.} \label{fig:capacity_analyses_withEs}
\end{figure*}

\textbf{Impact of Transmit and Receive Coil Numbers:} In Eqs.~\eqref{eq:COAM2_upper} and \eqref{eq:COAM2_lower}, $N_t$ outside the log provides a multiplexing gain while $I$ inside the log provides a diversity gain for the OAM-NFC capacity. Therefore, both $\widehat{C}^{\rm {upper}}_{\rm {OAM}}$ and $\widehat{C}^{\rm {lower}}_{\rm {OAM}}$ increase as $N_t$ and $N_r$ increase. Moreover, $\widehat{C}^{\rm {upper}}_{\rm {OAM}}$ and $\widehat{C}^{\rm {lower}}_{\rm {OAM}}$ increase faster as $N_{t}$ increases than $N_{r}$ increases. This is because $N_{t}$ corresponds to the number of available OAM modes while $N_{r}$ relates to the received SNR. Based on Eq.~\eqref{eq:C^LS_OAM}, Fig.~\ref{fig:Capacity_theo_N} gives the numerical results of different numbers of transmit and receive coils for the OAM-NFC capacity. In Fig.~\ref{fig:Capacity_theo_N}, we set $d_x\hspace{-0.1cm}=\hspace{-0.1cm}\theta_x\hspace{-0.1cm}=\hspace{-0.1cm}0$ with $N_t$ as well as $N_r$ varying from $1$ to $20$. We also set $R_t\hspace{-0.1cm}=\hspace{-0.1cm}R_r\hspace{-0.1cm}=\hspace{-0.1cm}50$ mm to avoid overlaps among coils. Other variables remain the same as those in Section~\ref{sec:capacity_analyses_LSvsW}. For $N_t\hspace{-0.1cm}\ne\hspace{-0.1cm}N_r$, since larger $N_t$ provides more OAM modes, the capacity increases as $N_t$ increases. For $N_t\hspace{-0.1cm}=\hspace{-0.1cm}N_r$, the OAM-NFC capacity increases as $N_t$ and $N_r$ increase depending on the increasing number of the orthogonal OAM modes.

\textbf{Impact of Transceiver Distance:} In Eq.~\eqref{eq:VKPsi2}, $\widehat k$ decreases as the transceiver distance increases, which leads to the decrease of $\Psi(\widehat k)$, thus decreasing the mutual inductance between each pair of transmit coil and receive coil. Therefore, increasing $D$ leads to the decrease of the OAM-NFC capacity in Eq.~\eqref{eq:COAM2}. Fig.~\ref{fig:Capacity_theo_D} plots the numerical results of different transceiver distances for the OAM-NFC capacity. In Fig.~\ref{fig:Capacity_theo_D}, we set $d_x\hspace{-0.1cm}=\hspace{-0.1cm}\theta_x\hspace{-0.1cm}=\hspace{-0.1cm}0$ with the value of $D$ varying from $0$ to $50$ mm. Other variables remain the same as Section~\ref{sec:capacity_analyses_LSvsW}. Fig.~\ref{fig:Capacity_theo_D} shows that the OAM-NFC capacity decreases as $D$ increases. However, the OAM-NFC capacity decreases more and more slowly as $D$ increases. This is because the crosstalk among different OAM modes also decreases as $D$ increases, which reduces the interferences among signals with different OAM modes.

\textbf{Impact of Transmit and Receive Coil Ring Radii:} In Eq.~\eqref{eq:VKPsi2}, for $R_t\hspace{-0.1cm}=\hspace{-0.1cm}R_r$, the transmit coil and receive coil of each pair are aligned with each other, which leads to ${d_{1,1}}\hspace{-0.1cm}=\hspace{-0.1cm}0$, thus maximizing the mutual inductance between each pair of transmit coil and receive coil. Therefore, the OAM-NFC capacity with $R_t\hspace{-0.1cm}=\hspace{-0.1cm}R_r$ is much higher than that with $R_t\hspace{-0.1cm}\ne\hspace{-0.1cm}R_r$. Also, ${d_{m,1}}$ increases while $\widehat V$ and $\widehat k$ decrease as $R_t$ and $R_r$ increase, which leads to the decrease of mutual inductance, thus decreasing the OAM-NFC capacity. Fig.~\ref{fig:Capacity_theo_R} plots the numerical results of different transmit and receive coil ring radii for the OAM-NFC capacity. In Fig.~\ref{fig:Capacity_theo_R}, we set $d_x\hspace{-0.1cm}=\hspace{-0.1cm}\theta_x\hspace{-0.1cm}=\hspace{-0.1cm}0$ with the values of $R_t$ as well as $R_r$ varying from $20$ to $100$ mm. Other variables remain the same as those in Section~\ref{sec:capacity_analyses_LSvsW}. Fig.~\ref{fig:Capacity_theo_R} verifies that the OAM-NFC capacity with $R_t\hspace{-0.1cm}=\hspace{-0.1cm}R_r$ is much higher than that with $R_t\hspace{-0.1cm}\ne\hspace{-0.1cm}R_r$. This is because the same transmit and receive coil ring radius
maximizes the mutual inductance. Also, we can find that the OAM-NFC capacity decreases as $R_t$ and $R_r$ increase for $R_t\hspace{-0.1cm}=\hspace{-0.1cm}R_r$.

\textbf{Impact of Transmit and Receive Coil Radii:} In Eq.~\eqref{eq:M}, for aligned transceivers with the same coil ring radii, we have $d_{m,n}\hspace{-0.1cm}=\hspace{-0.1cm}0$, which leads to $2d_{m,n}r_t\sin t\hspace{-0.1cm}=\hspace{-0.1cm}0$. Therefore, $M_{m,n}$ is positively correlated with $r_t$ and $r_r$. Consequently, the mutual inductance increases as $r_t$ and $r_r$ increase, thus leading to a higher OAM-NFC capacity in Eqs.~\eqref{eq:COAM} and \eqref{eq:COAM2}. Fig.~\ref{fig:Capacity_theo_rr} plots the numerical results of different transmit and receive coil radii for the OAM-NFC capacity. In Fig.~\ref{fig:Capacity_theo_rr}, we set $d_x\hspace{-0.1cm}=\hspace{-0.1cm}\theta_x\hspace{-0.1cm}=\hspace{-0.1cm}0$ with the values of $r_t$ as well as $r_r$ varying from $0$ to $15$ mm. We also set $R_t\hspace{-0.1cm}=\hspace{-0.1cm}R_r\hspace{-0.1cm}=\hspace{-0.1cm}15$ mm to avoid overlaps among coils. Other variables remain the same as those in Section~\ref{sec:capacity_analyses_LSvsW}. Fig.~\ref{fig:Capacity_theo_rr} shows that the OAM-NFC capacity increases as $r_t$ and $r_r$ increase. However, this increase slows down as $r_t$ and $r_r$ increase. This is because as $r_t$ and $r_r$ become closer and closer to $R_t$ and $R_r$, respectively, the crosstalks increase, thus increasing the interferences among signals with different OAM modes. Also, due to $R = 2\pi r_tK_t R_0 /S$ and $Z_t=R+{1}/{j\omega C}+j\omega L_t$, the increase of $r_t$ leads to the increase of $Z_t$, thus decreasing the capacity. Therefore, the OAM-NFC capacity increases faster as $r_r$ increases than $r_t$.

\subsection{SISO-NFC and MIMO-NFC Channel Capacity With Channel Estimation}
For comparison, we also give the capacities of conventional SISO-NFC and MIMO-NFC. First, we give the induced voltage of SISO-NFC, denoted by $V^i\hspace{-0.1cm}=\hspace{-0.1cm}\frac{-j\omega M'}{Z_t}V^s$,
where $V^s$ represents the source voltage with $\left|V^s\right|^2\hspace{-0.1cm}=\hspace{-0.1cm}P_t$ and $M'$ represents the mutual induction between the transmit coil and receive coil ($M'$ is equal to $M_{m,n}$ with $n\hspace{-0.1cm}=\hspace{-0.1cm}m$). Without loss of generality, let $M'\hspace{-0.1cm}=\hspace{-0.1cm} M_{1,1}$. Then, the capacity of SISO-NFC, denoted by $C_{\rm {SISO}}$, can be given as follows:
\begin{align}
{C}_{\rm {SISO}}&={\rm{log}}\left(1+\frac{\left|V^i\right|^2}{N_0}\right)\hspace{-0.1cm}
=\hspace{-0.1cm}{\rm{log}}\left(1+\frac{P_t}{N_0}\frac{\omega^2\left| M_{1,1}\right|^2}{\left|Z_t\right|^2}\right).
\label{eq:CSISO}
\end{align}

Second, we give the capacity of MIMO-NFC with channel estimation. With the same transmit and receive coil rings for transmission, the complex channel gain of MIMO-NFC is equal to $\boldsymbol{\mathrm H}$. However, since the size of each coil is very small compared with the wavelength, the crosstalks among transmit coils are strong, MIMO-NFC channel is highly correlated\cite{oam_MIMO}. Thus, the spatial correlation should be taken into consideration. We denote by $\boldsymbol{\rm G}_t$ and $\boldsymbol{\rm G}_r$ the correlation matrices at the transmit and the receive coils, respectively. Thus, the channel matrix of MIMO-NFC can be given by $\boldsymbol{\rm H}_{\rm {MIMO}}=\boldsymbol{\rm G}_{r} \boldsymbol{\rm H} \boldsymbol{\rm G}_{t}$. Thus, the channel matrix of MIMO-NFC can be given by $\boldsymbol{\rm H}_{\rm {MIMO}}=\boldsymbol{\rm G}_{r} \boldsymbol{\rm H} \boldsymbol{\rm G}_{t}$, and the channel capacity of MIMO-NFC, denoted by $C_{\rm {MIMO}}$, can be given as follows:
\begin{align}
C_{\rm {MIMO}}&\hspace{-0.1cm}=\hspace{-0.1cm}{\rm{log}}\left[{\rm det}\left(\boldsymbol{\rm I}+\frac{P_t}{N_0N_t}\boldsymbol{\rm H}_{\rm {MIMO}}\boldsymbol{\rm H}_{\rm {MIMO}}^H\right)\right]
\nonumber\\ &\hspace{-0.1cm}=\hspace{-0.1cm}{\rm{log}}\left[{\rm det}\left(\boldsymbol{\rm I}+\frac{P_t}{N_0N_t}\boldsymbol{\rm G}_{r} \boldsymbol{\rm H} \boldsymbol{\rm G}_{t} \boldsymbol{\rm G}^H_{t} \boldsymbol{\rm H}^H \boldsymbol{\rm G}^H_{r}\right)\right],
\label{eq:CMIMO}
\end{align}
where $\boldsymbol{\rm I}$ is the unit matrix. What should be noticed is that the capacity of MIMO-NFC achieved by using SVD, smart antennas, or water-filling power allocation needs the transmitter to know the CSI, which requires the feedback from the receiver to the transmitter. Although simple channel estimation methods can be combined with NFC by using digital signal processing at the receiver, whether CSI feedback can be deployed for multi-coil NFC systems has not been confirmed. This is because NFC is based on the resonant inductive coupling between the transmit and receive coils. The CSI feedback, which can be obtained in SISO-NFC by analyzing the change of current in the transmit coil, is not available for multi-coil NFC system due to the interference from the transmit coils as well as other receiver coils. Also, if the CSI is sent to the transmitter by using an additional feedback coil pair, not only will it causes new interference, but also will add more hardware cost.

\subsection{Impact of Transmit and Receive Coil Numbers on MIMO-NFC and OAM-NFC}
\begin{figure}[htbp]
\centering
\includegraphics[scale=0.6]{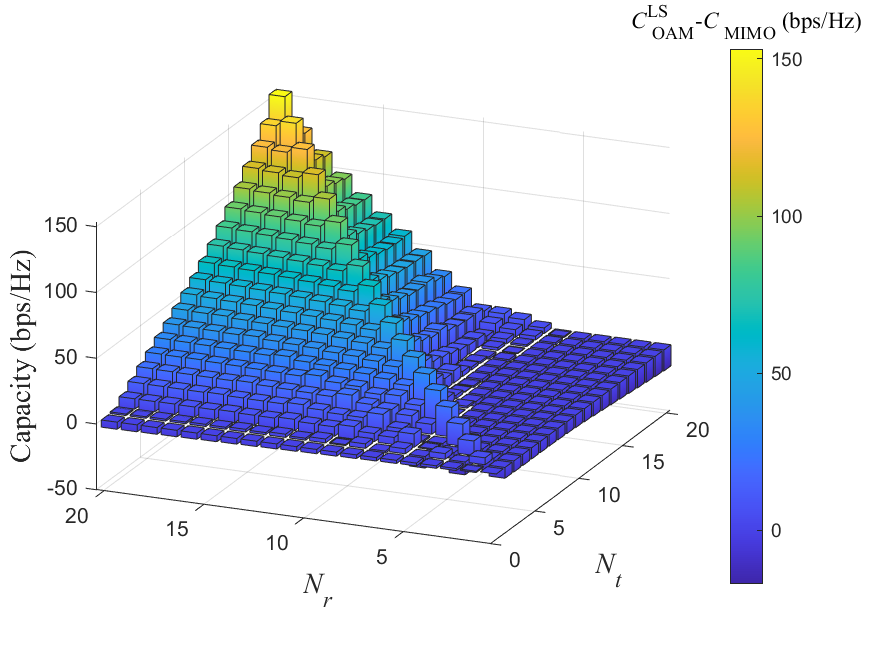}
\caption{Capacity gap between OAM-NFC and MIMO-NFC with different transmit and receive coil numbers.}
\label{fig:COAM-CMIMO}
\end{figure}
Using channel estimation for OAM-NFC at the receiver, the input signals can be detected for $N_r\hspace{-0.1cm}\ne\hspace{-0.1cm} IN_t$. This is because the multi-coil mutual inductance channel can be transformed to a square matrix, making the irregular coils do not decrease the OAM-NFC capacity. We also analyze whether MIMO-NFC outperforms OAM-NFC with $N_r\hspace{-0.1cm}\ne\hspace{-0.1cm} IN_t$ by giving the capacity gap between OAM-NFC and MIMO-NFC in Fig.~\ref{fig:COAM-CMIMO}. We set $d_x\hspace{-0.1cm}=\hspace{-0.1cm}\theta_x\hspace{-0.1cm}=\hspace{-0.1cm}0$ and set $N_t$ as well as $N_r$ vary from $1$ to $20$. We also set $R_t\hspace{-0.1cm}=\hspace{-0.1cm}R_r\hspace{-0.1cm}=\hspace{-0.1cm}50$ mm to avoid overlaps among coils. Other variables remain the same as those in Section IV-C. Fig.~\ref{fig:COAM-CMIMO} shows that, although MIMO-NFC has a slightly higher capacity than OAM-NFC with small $N_t$ or $N_r$, OAM-NFC has a much higher capacity than MIMO-NFC with large $N_t$ and $N_r$. This is because, for $R_t\hspace{-0.1cm}=\hspace{-0.1cm}R_r\hspace{-0.1cm}=\hspace{-0.1cm}50$, small $N_t$ or $N_r$ leads to a non-dense placed coil structure and decreases the correlation among coils, thus increasing the MIMO-NFC capacity. Therefore, the irregular condition does not provide superiority in MIMO-NFC systems for densely placed coils.

\section{Verification and Simulation Results}\label{sec:Simulation}
In this section, we first give our proposed OAM-NFC simulation model in ANSYS high frequency structure simulator (HFSS). Then, we explain the relationship between the S-parameter matrix and the channel matrix. After that, we validate the feasibility of OAM-NFC by analyzing the phase structures and BERs, respectively. Also, we compare the OAM-NFC capacity with those of SISO-NFC and MIMO-NFC to validate the capacity enhancement. How imperfect alignment, the transceiver coil number, transceiver distance, transceiver coil ring radii, and coil radii, impact the capacity, respectively, are also evaluated.

\subsection{HFSS Model}\label{sec:HFSS_model}
\begin{figure}[htbp]
\centering
\subfigure[The OAM-NFC coil structure.]{
\begin{minipage}{1\linewidth}
\centering
\includegraphics[scale=0.53]{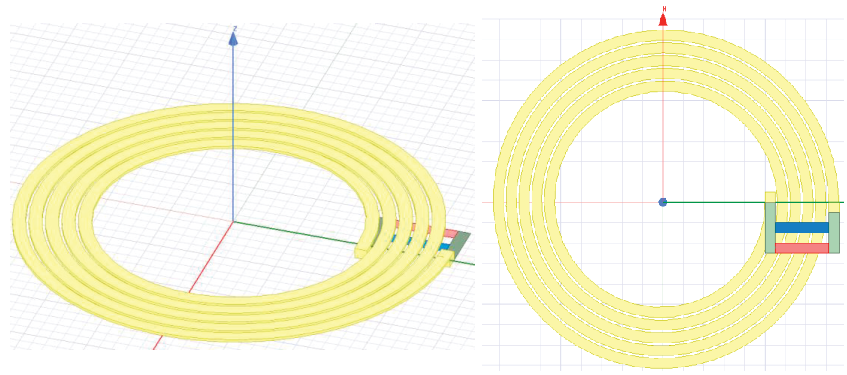}
\label{fig:HFSScoil}
\end{minipage}
}\\
\subfigure[Transceiver coil rings.]{
\begin{minipage}{1\linewidth}
\centering
\includegraphics[scale=0.56]{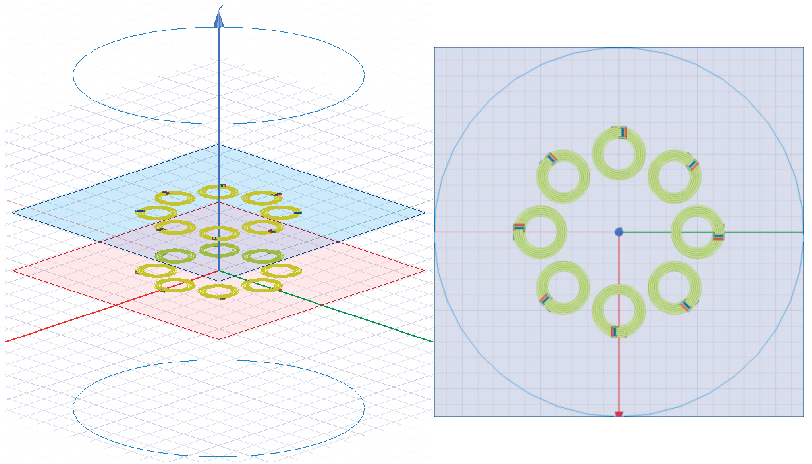}
\label{fig:HFSSUCA}
\end{minipage}
}
\centering
\caption{The OAM-NFC simulation model in HFSS.} \label{fig:HFSS_ model}
\end{figure}
Figure~\ref{fig:HFSS_ model} shows the simulation models of the coil and transceiver  coil rings. All the transmit and the receive coils share the same model as shown in Fig.~\ref{fig:HFSScoil}. We set the coil material as copper, the line width of each coil as $0.5$ mm, and the thickness as $0.1$ mm. We also set $K_t\hspace{-0.1cm}=\hspace{-0.1cm}K_r\hspace{-0.1cm}=\hspace{-0.1cm}5$ and $r_t\hspace{-0.1cm}=\hspace{-0.1cm}r_r\hspace{-0.1cm}=\hspace{-0.1cm}5$ mm. The matching parallel capacitance is set as $300$ pF with a resistance of $50 $ $\Omega$. In Fig.~\ref{fig:HFSSUCA}, we set $R_t\hspace{-0.1cm}=\hspace{-0.1cm}R_r\hspace{-0.1cm}=\hspace{-0.1cm}25$ mm, $N_t\hspace{-0.1cm}=\hspace{-0.1cm}N_r\hspace{-0.1cm}=\hspace{-0.1cm}8$, and $D$ as $25$ mm. The S-parameter matrix of ports, which represents the output to input voltage ratios across ports, is given as follows:
\small{
\begin{align}
\boldsymbol{\mathrm S}=
\begin{bmatrix}
S_{1,1} & S_{1,j} & \cdots & S_{1,N_r+N_t}\\
S_{2,1} & S_{2,j} & \cdots & S_{2,N_r+N_t}\\
\vdots & \vdots & \ddots & \vdots \\
S_{i,1} & S_{i,j} & \cdots & S_{i,N_r+N_t}\\
\vdots & \vdots & \ddots & \vdots \\
S_{N_r+N_t,1} & S_{N_r+N_t,j} & \cdots & S_{N_r+N_t,N_r+N_t}
\end{bmatrix},
\end{align}
}
where $S_{i,j}={V_j}/{V_i}$ with $V_i$ and $V_j$ ($1\hspace{-0.1cm}\le\hspace{-0.1cm} i,j\hspace{-0.1cm}\le\hspace{-0.1cm} N_r\hspace{-0.07cm}+\hspace{-0.07cm}N_t$) represent the voltages across ports. What needs to be noticed is that $\boldsymbol{\mathrm S}$ gives the relationships among all ports, including those between each two transmit coils and the return loss. Thus, $\boldsymbol{\mathrm H}$ is the bottom left quarter of $\boldsymbol{\mathrm S}$, and is rewritten as follows:
\begin{align}
\boldsymbol{\mathrm H}=
\begin{bmatrix}
S_{N_t+1,1}
& S_{N_t+1,2} & \cdots & S_{N_t+1,N_t}\\
S_{N_t+2,1}
& S_{N_t+2,2} & \cdots & S_{N_t+2,N_t}
\\ \vdots & \vdots & \ddots & \vdots \\ S_{N_t+N_r,1}
& S_{N_t+N_r,2} & \cdots & S_{N_t+N_r,N_t}
\end{bmatrix}.
\label{eq:H_S_relation}
\end{align}

\subsection{Verification of Feasibility}\label{sec:Verification}
We validate the feasibility of our proposed OAM-NFC system by simulating the phase structures and BERs. The phase structures at $5.8$ GHz and $13.56$ MHz have the same trend. First, we set the frequency as $5.8$ GHz to get more clear phase structures. Then, we set the frequency as $13.56$ MHz since this frequency is the most commonly used in NFC.

\begin{figure}[htbp]
\centering
\subfigure[OAM mode $l$=2.]{
\begin{minipage}{0.45\linewidth}
\centering
\includegraphics[scale=0.145]{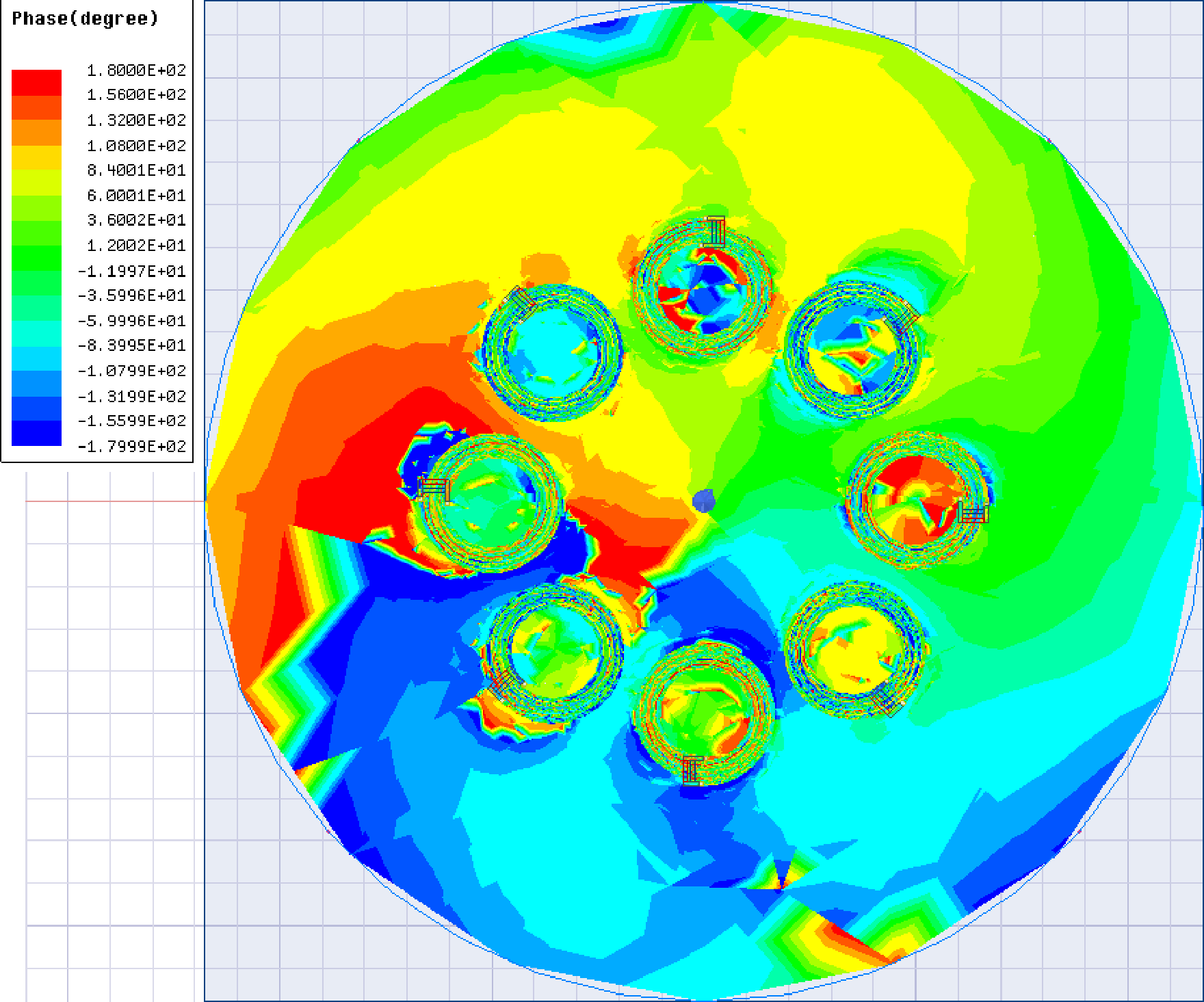}
\label{fig:phase_structure_2_2}
\end{minipage}
}
\subfigure[OAM mode $l$=3.]{
\begin{minipage}{0.45\linewidth}
\centering
\includegraphics[scale=0.145]{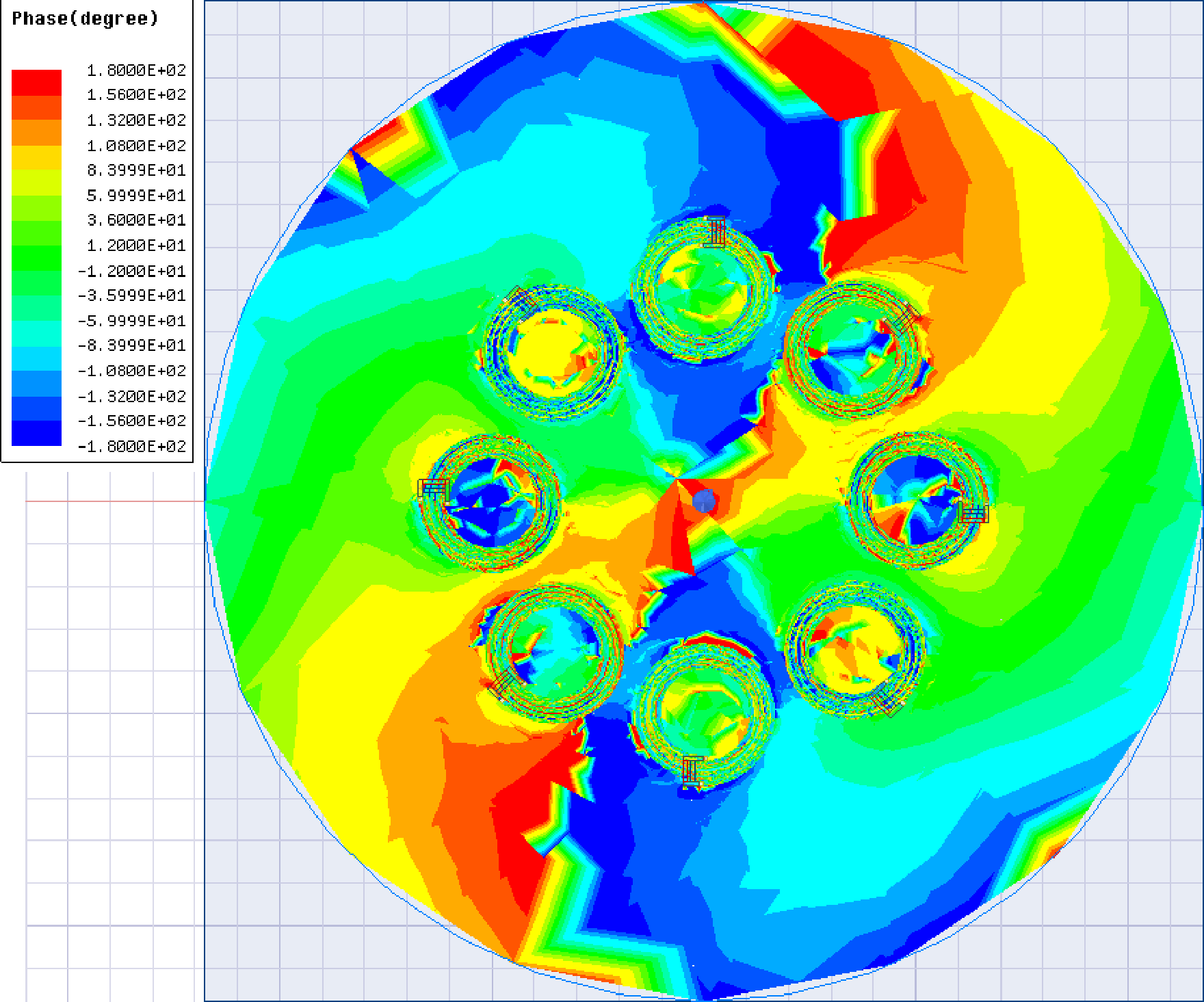}
\label{fig:phase_structure_2_3}
\end{minipage}
}
\centering
\caption{Phase structures of OAM-NFC with $l\hspace{-0.1cm}=\hspace{-0.1cm}2$ and $l\hspace{-0.1cm}=\hspace{-0.1cm}3$ for $D\hspace{-0.1cm}=\hspace{-0.1cm}25$ mm at $5.8$ GHz.} \label{fig:phase_structures}
\end{figure}

\begin{figure}[htbp]
\centering
\subfigure[OAM mode $l$=2.]{
\begin{minipage}{0.45\linewidth}
\centering
\includegraphics[scale=0.145]{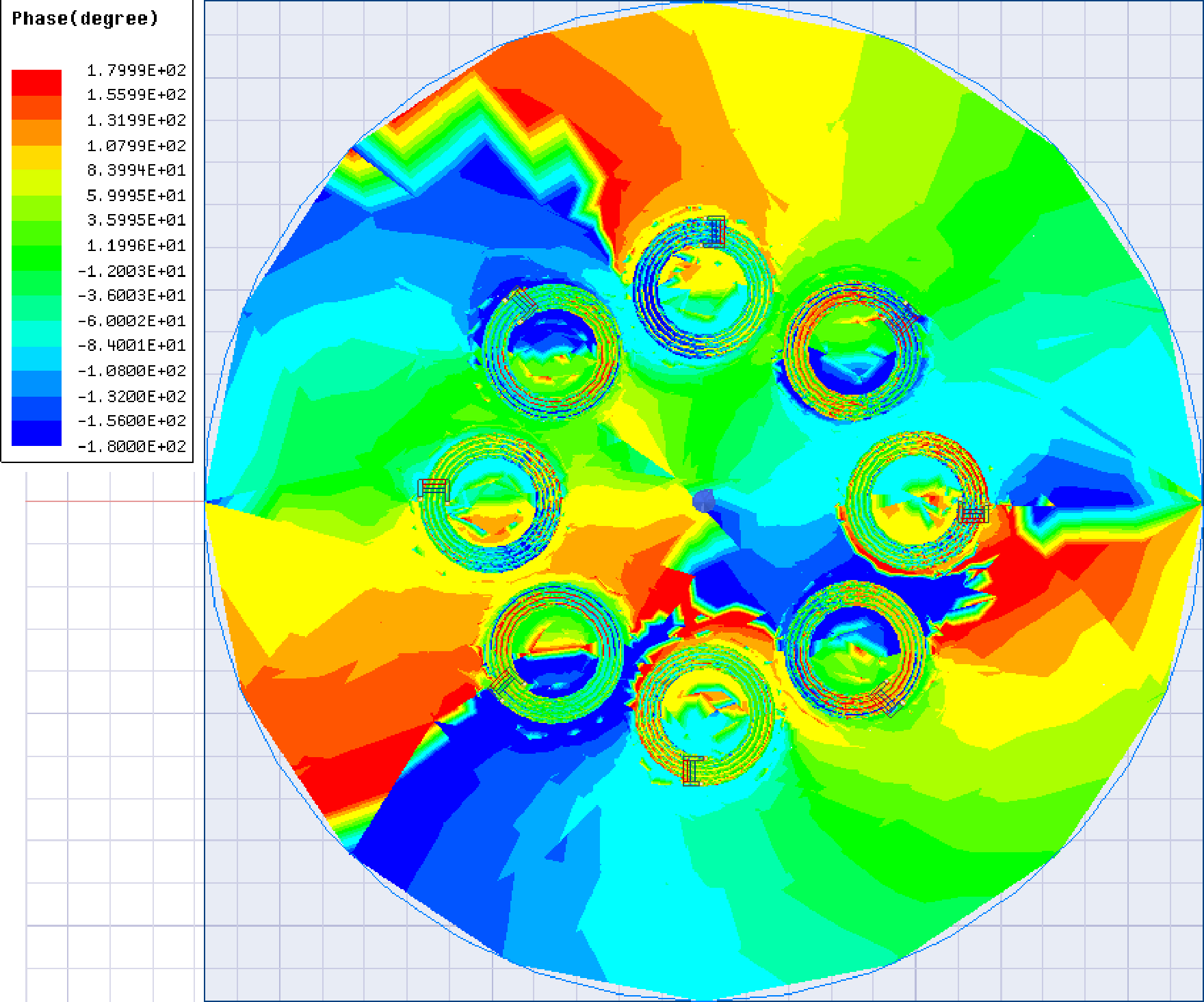}
\end{minipage}
}
\subfigure[OAM mode $l$=3.]{
\begin{minipage}{0.45\linewidth}
\centering
\includegraphics[scale=0.145]{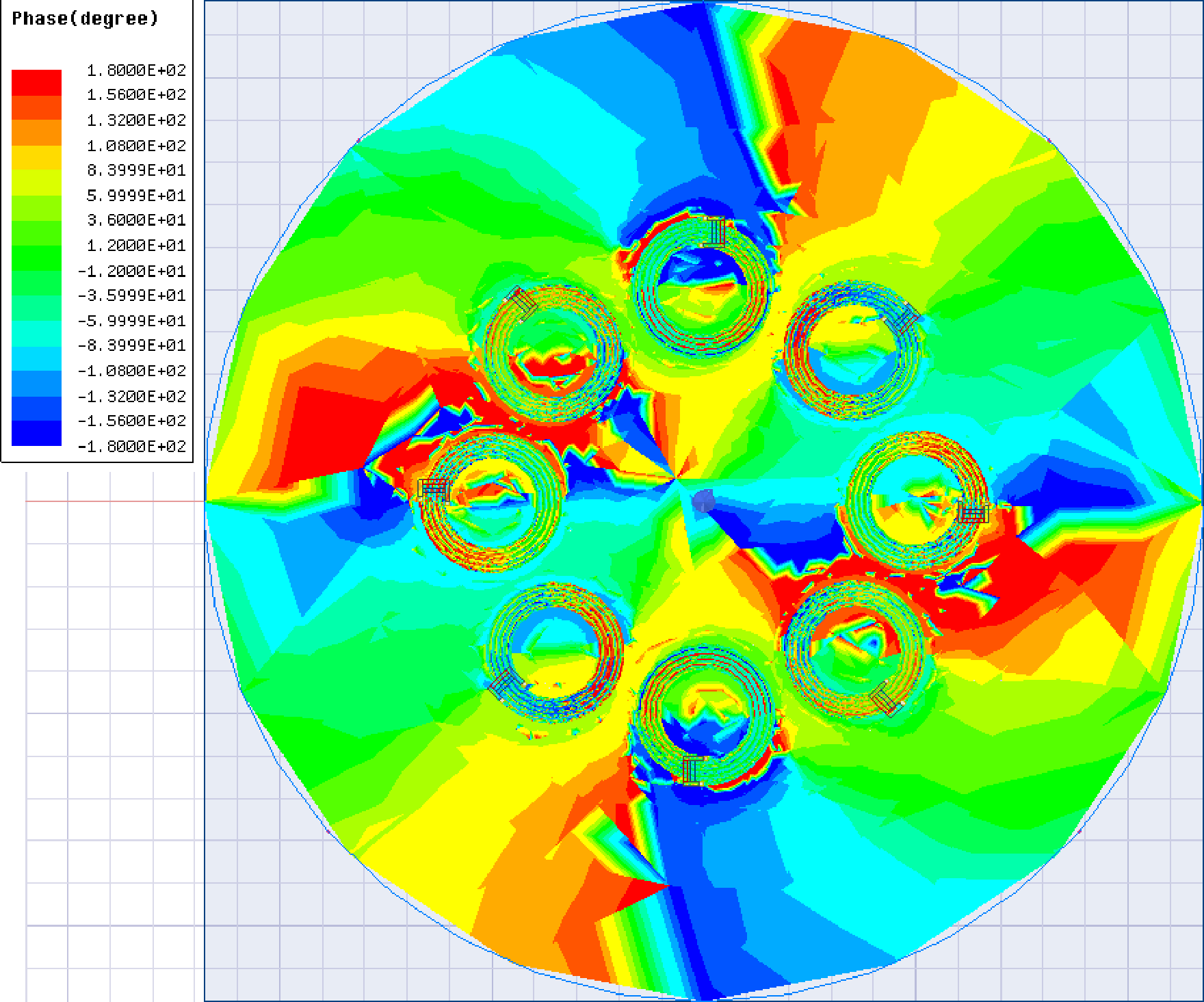}
\end{minipage}
}
\centering
\caption{Phase structures of OAM-NFC with $l\hspace{-0.1cm}=\hspace{-0.1cm}2$ and $l\hspace{-0.1cm}=\hspace{-0.1cm}3$ for $D\hspace{-0.1cm}=\hspace{-0.1cm}25$ mm at $13.56$ MHz.} \label{fig:phase_structures_1356}
\end{figure}
For $5.8$ GHz, the phase structures on the receive coil ring plane with $l\hspace{-0.1cm}=\hspace{-0.1cm}2$ and $l\hspace{-0.1cm}=\hspace{-0.1cm}3$ are shown in Fig.~\ref{fig:phase_structures}. The variation of color from blue to red, yellow, green, and back to blue corresponds to the change in phase of $2\pi$. The rotational phase structure, which is a typical feature of the OAM beam, can be clearly found in Fig.~\ref{fig:phase_structures}. The existence of rotational phase structures validates the feasibility of our proposed OAM-NFC system at $5.8$ GHz. For $13.56$ MHz, the typical rotational phase structures can still be found within the coil ring in Fig.~\ref{fig:phase_structures_1356}, validating that our proposed OAM-NFC system is feasible for aligned transceivers with $N_t\hspace{-0.1cm}=\hspace{-0.1cm}N_r$ at $5.8$ GHz and $13.56$ MHz.

\begin{figure}[htbp]
\centering
\subfigure[5.8 GHz.]{
\begin{minipage}{0.45\linewidth}
\centering
\includegraphics[scale=0.145]{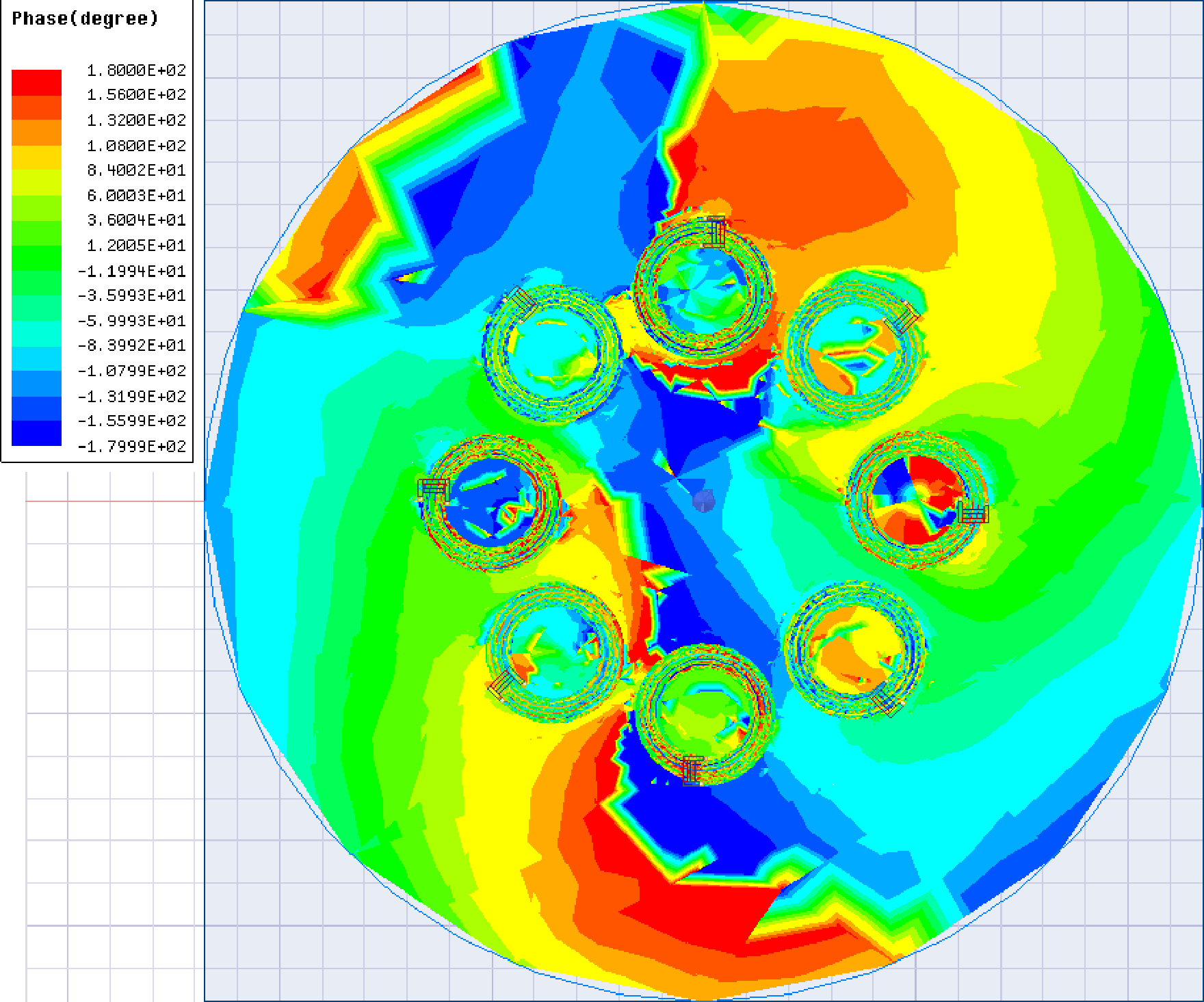}
\end{minipage}
}
\subfigure[13.56 MHz.]{
\begin{minipage}{0.45\linewidth}
\centering
\includegraphics[scale=0.145]{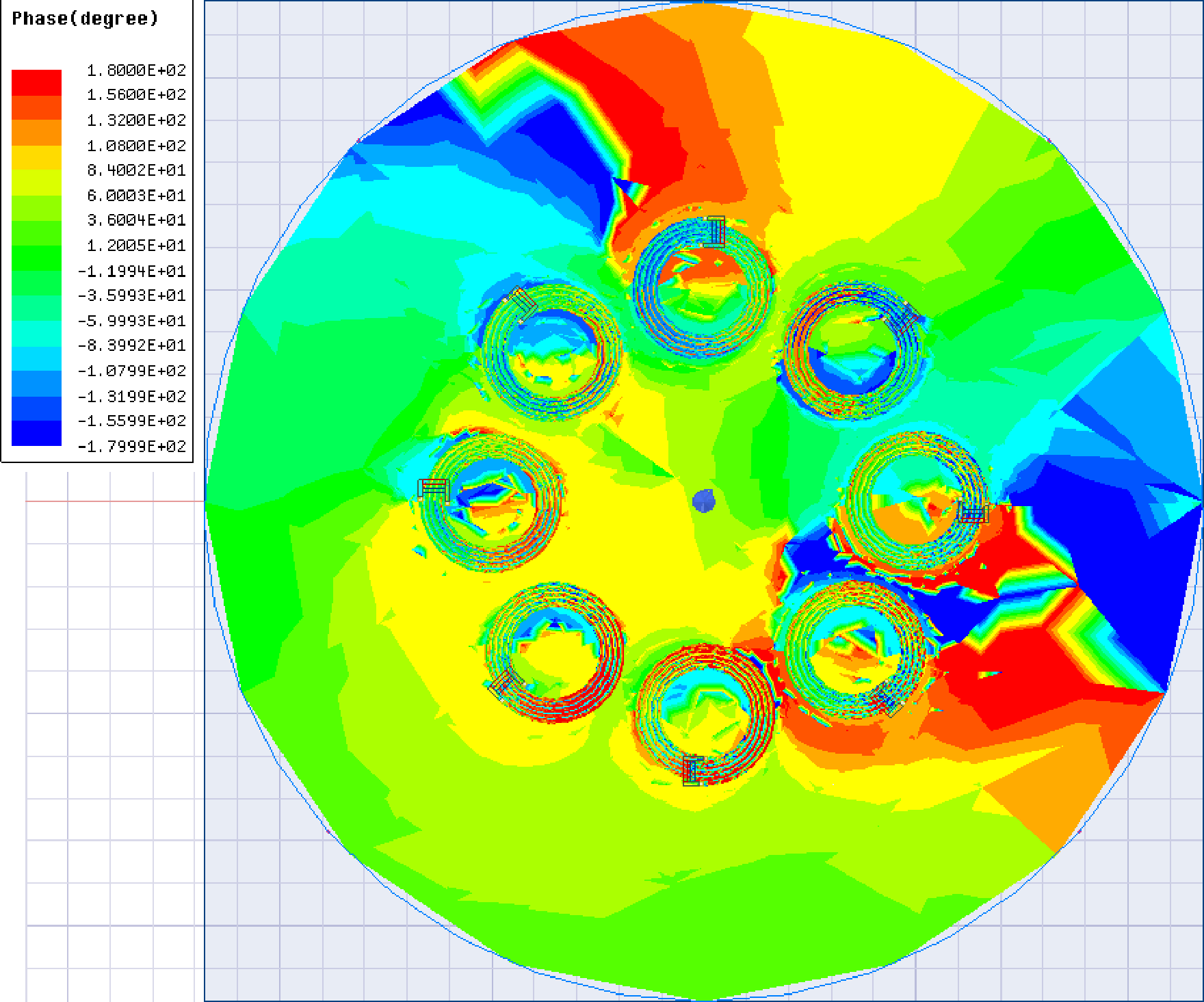}
\end{minipage}
}
\centering
\caption{Mixed phases of $l\hspace{-0.1cm}=\hspace{-0.1cm}2$ with $l\hspace{-0.1cm}=\hspace{-0.1cm}3$ for $D\hspace{-0.1cm}=\hspace{-0.1cm}25$ mm at $5.8$ GHz and $13.56$ MHz.} \label{fig:phase_structures_mixed}
\end{figure}
\begin{figure}[htbp]
\centering
\includegraphics[scale=0.6]{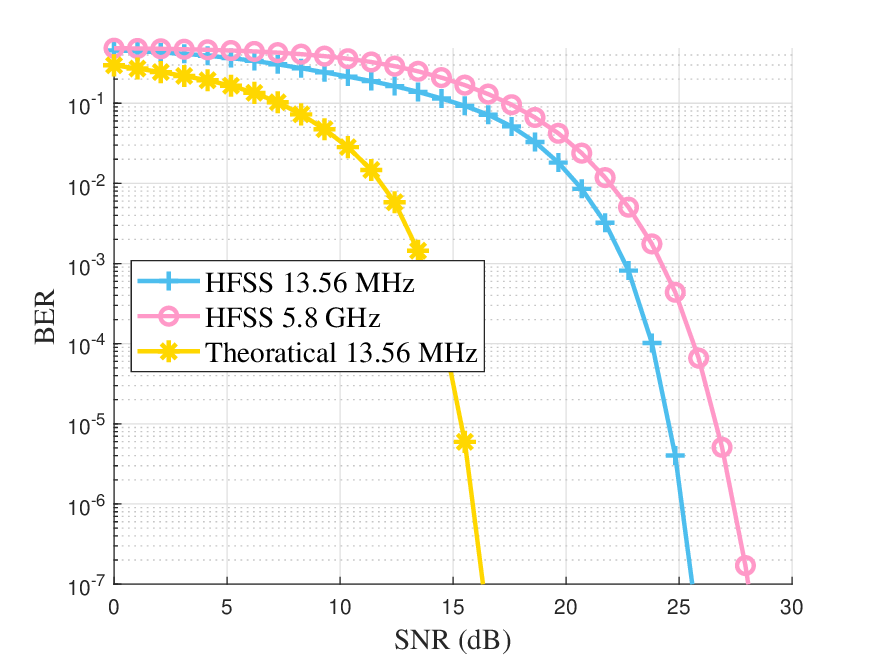}
\caption{BERs of OAM-NFC at $5.8$ GHz and $13.56$ MHz.}
\label{fig:BER}
\end{figure}
\begin{figure}[htbp]
\centering
\includegraphics[scale=0.6]{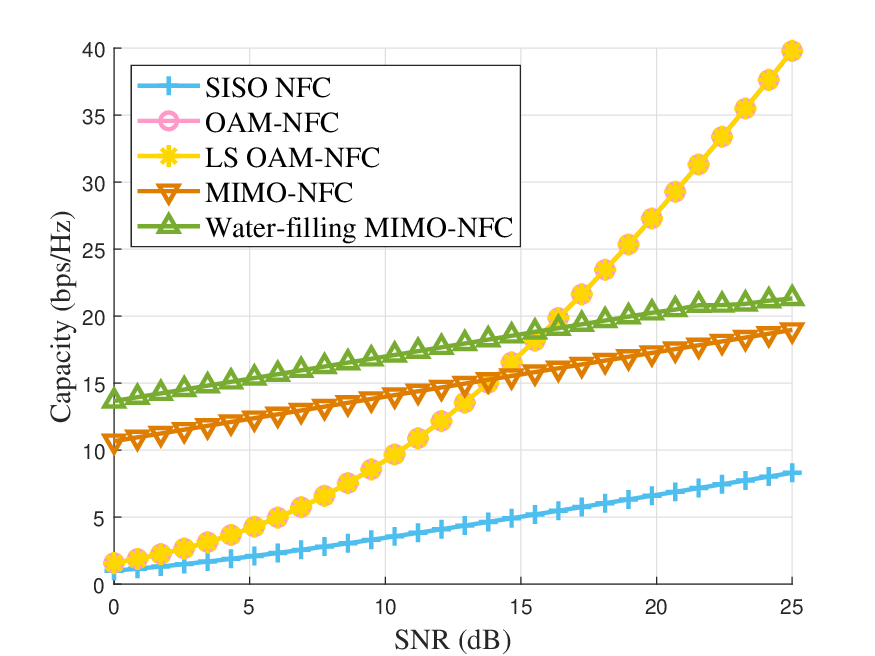}
\caption{Capacity comparison for OAM-NFC, SISO-NFC, and MIMO-NFC.}
\label{fig:capacity}
\end{figure}
\begin{figure*}[htbp]
\centering
\subfigure[Different deflection angles without channel estimation.]{
\begin{minipage}{0.45\linewidth}\label{fig:Capacity_HFSS_deg_noes}
\centering
\includegraphics[scale=0.55]{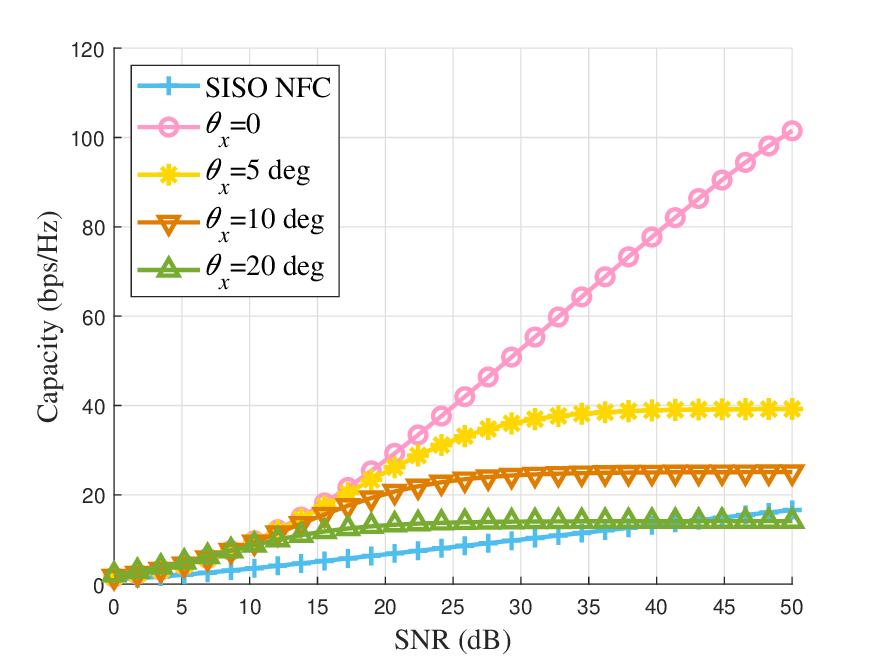}
\end{minipage}
}
\subfigure[Different deflection angles with channel estimation.]{
\begin{minipage}{0.45\linewidth}\label{fig:capacity_deg}
\centering
\includegraphics[scale=0.55]{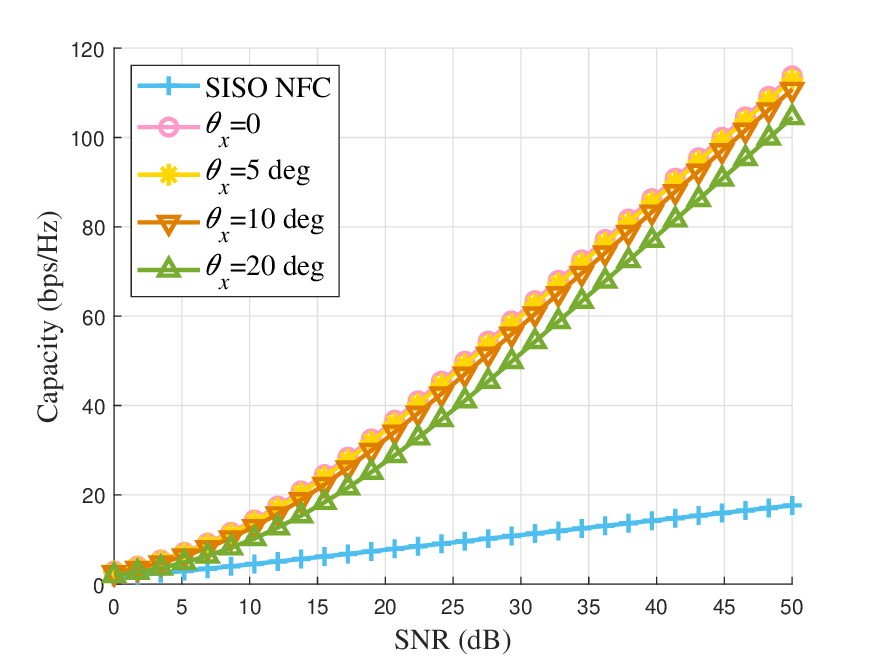}
\end{minipage}
}\\
\subfigure[Different deflection distances without channel estimation.]{
\begin{minipage}{0.45\linewidth}\label{fig:capacity_dx_noes}
\centering
\includegraphics[scale=0.55]{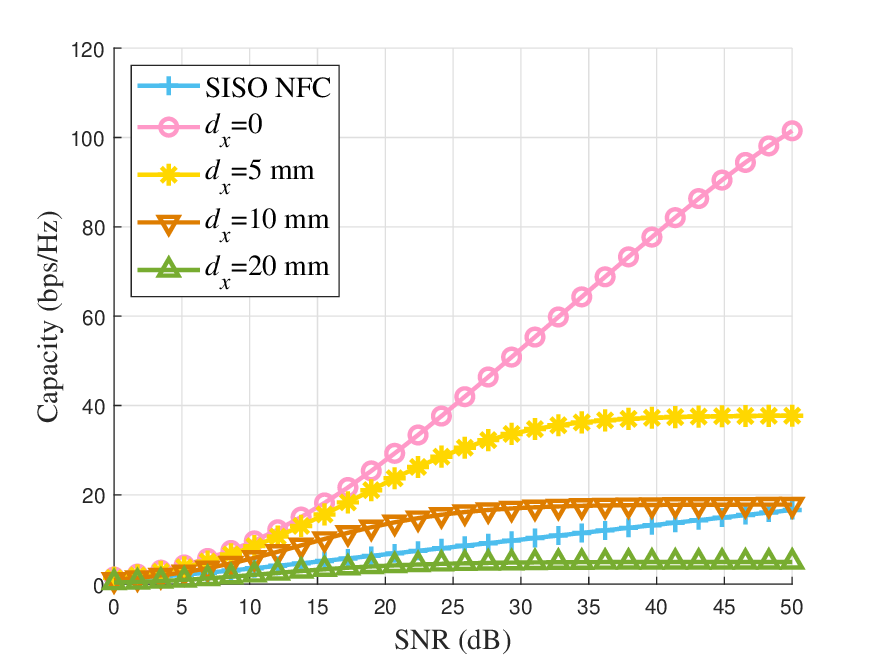}
\end{minipage}
}
\subfigure[Different deflection distances with channel estimation.]{
\begin{minipage}{0.45\linewidth}\label{fig:capacity_dx}
\centering
\includegraphics[scale=0.55]{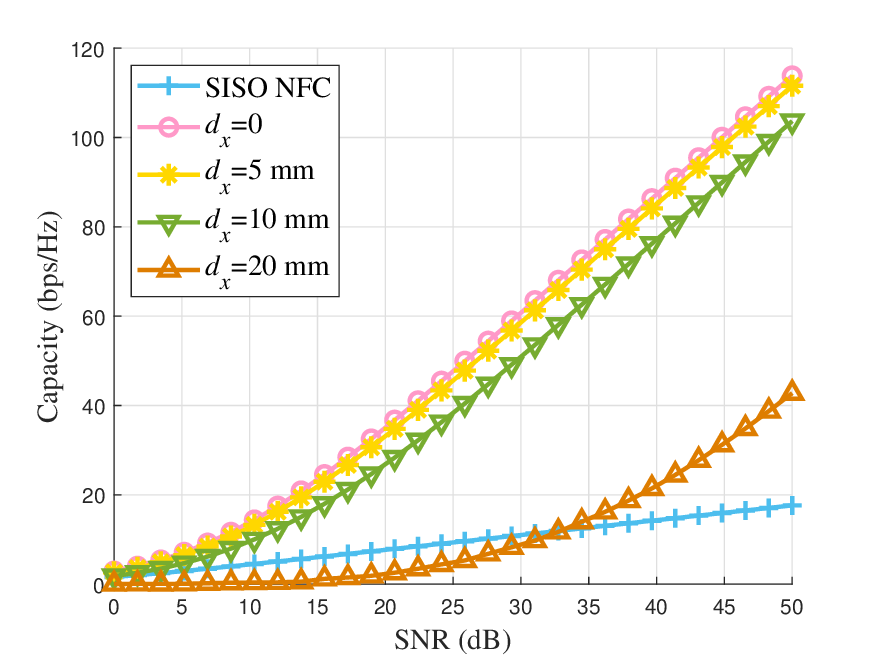}
\end{minipage}
}
\centering
\caption{The impacts of transceiver misalignment on the OAM-NFC capacity.}\label{fig:capacity_misalign}
\end{figure*}
Figure~\ref{fig:phase_structures_mixed} shows the mixed phases of $l\hspace{-0.1cm}=\hspace{-0.1cm}2$ with $l\hspace{-0.1cm}=\hspace{-0.1cm}3$ at $5.8$ GHz and $13.56$ MHz, respectively. The phase structures are irregular and BER analysis is used to validate the OAM-NFC feasibility. Fig.~\ref{fig:BER} plots BERs of OAM-NFC at $5.8$ GHz and $13.56$ MHz without channel estimation. In Fig.~\ref{fig:BER}, we set $P_t$ as $8$ W and let $P_n\hspace{-0.1cm}=\hspace{-0.1cm}P_t/$SNR with SNR varying from $0$ to $30$ dB. Fig.~\ref{fig:BER} shows that OAM-NFC with frequency at $13.56$ MHz has better BER performance than that at $5.8$ GHz. This is because the resonance points of the coils are set at $13.35$ MHz. For SNR larger than $17$ dB, the BER of theoretical results at $13.56$ MHz is less than $10^{-7}$. For SNR larger than $28$ dB, BERs of HFSS simulated results at both $13.56$ MHz and $5.8$ GHz are also less than $10^{-7}$. The BER simulation results validate the feasibility of our proposed OAM-NFC system. 

\subsection{Verification of Channel Capacity Enhancement}
Figure~\ref{fig:capacity} plots the OAM-NFC capacity as well as those of SISO-NFC and MIMO-NFC simulated by HFSS. In Fig.~\ref{fig:capacity}, we set the frequency as $13.56$ MHz and let SNR vary from $0$ to $25$ dB. Other variables remain the same as those in Section \ref{sec:HFSS_model}. Fig.~\ref{fig:capacity} shows that OAM-NFC has the same capacity with and without channel estimation under the aligned condition. Also, it can be found that the OAM-NFC capacity is much higher than those of SISO-NFC and MIMO-NFC with and without water-filling power allocation for SNR $>15$ dB. This is because, although the diversity gain achieved by MIMO technology makes the capacity of MIMO-NFC higher than that of OAM-NFC in the low-SNR regime, with SNR increased, OAM-NFC outperforms MIMO-NFC due to its higher multiplexing gain. Since NFC is usually used in high SNR scenarios, our proposed OAM-NFC can significantly enhance the capacity of NFC systems.

\subsection{Impact of Misalignment}
In Fig.~\ref{fig:capacity_misalign}, we simulate how transceiver misalignment impacts the OAM-NFC capacity. Figs.~\ref{fig:Capacity_HFSS_deg_noes} and \ref{fig:capacity_deg} show how non-parallel transceivers impact the OAM-NFC capacity with $\theta_x\hspace{-0.1cm}=\hspace{-0.1cm}\theta_y\hspace{-0.1cm}=\hspace{-0.1cm}0$, $5$, $10$, and $20$ degrees while the value of SNR varies from $0$ to $50$ dB. Other variables remain the same as those in Section~\ref{sec:HFSS_model}. As illustrated in Fig.~\ref{fig:Capacity_HFSS_deg_noes}, the OAM-NFC capacity without channel estimation sharply decreases as $\theta$ increases, which is consistent with the numerous analyses of Fig.~\ref{fig:Capacity_theo_W_20dB} in Section~\ref{sec:capacity_analyses}. Moreover, OAM-NFC with a larger $\theta$ obtains its maximum capacity in a lower SNR than that with a smaller $\theta$. For SNR $>40$ dB as well as $\theta\hspace{-0.1cm}=\hspace{-0.1cm}20$ degrees, SISO-NFC has a larger capacity than the OAM-NFC without channel estimation. This is because the interferences among different OAM modes, which are caused by the non-parallel between transmit and receive coil rings, increase as $\theta$ increases. Fig.~\ref{fig:capacity_deg} shows that the non-parallel transceivers only reduce the OAM-NFC capacity slightly, which is consistent with the numerous analyses of Fig.~\ref{fig:Capacity_theo_LS_20dB} and verifies the robustness of our proposed OAM-NFC.
Figs.~\ref{fig:capacity_dx_noes} and \ref{fig:capacity_dx} show how non-coaxial transceivers impact the OAM-NFC capacity with $d_x\hspace{-0.1cm}=\hspace{-0.1cm}d_y\hspace{-0.1cm}=\hspace{-0.1cm}0$, $5$, $10$, and $20$ mm. Other variables remain the same as those in Section~\ref{sec:HFSS_model}. In Fig.~\ref{fig:capacity_dx_noes}, the capacity of OAM-NFC without channel estimation decreases more sharply as $d_x$ increases than as $\theta$ increases. For SNR $>50$ dB as well as $d_x\hspace{-0.1cm}=\hspace{-0.1cm}20$ mm, SISO-NFC has a larger capacity than the OAM-NFC without channel estimation. Fig.~\ref{fig:capacity_dx} shows that except for $d_x\hspace{-0.1cm}=\hspace{-0.1cm}20$ mm, the non-coaxial transceivers only reduce the OAM-NFC capacity slightly. Simulation results in Figs.~\ref{fig:capacity_dx_noes} and \ref{fig:capacity_dx} are consistent with the analyses in Section~\ref{sec:capacity_analyses}, verifying that channel estimation can significantly improve the robustness of OAM-NFC against the misalignment.

\begin{figure}[htbp]
\centering
\subfigure[OAM-NFC capacities with different numbers of transmit and receive coils.]{
\begin{minipage}{1\linewidth}\label{fig:capacity_N_M}
\centering
\includegraphics[scale=0.6]{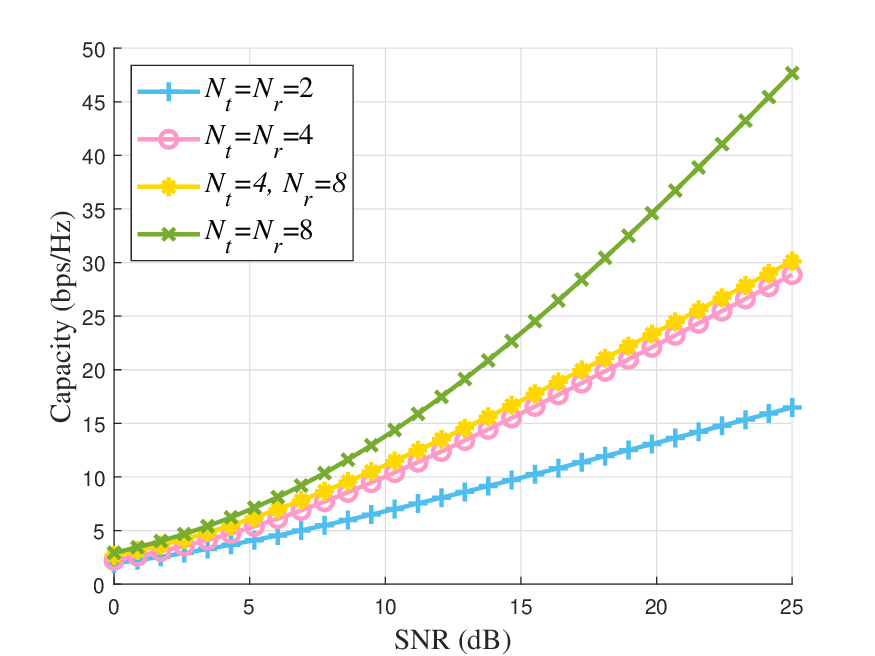}
\end{minipage}
}\\
\subfigure[OAM-NFC capacities with different transceiver distances.]{
\begin{minipage}{1\linewidth}\label{fig:capacity_D}
\centering
\includegraphics[scale=0.6]{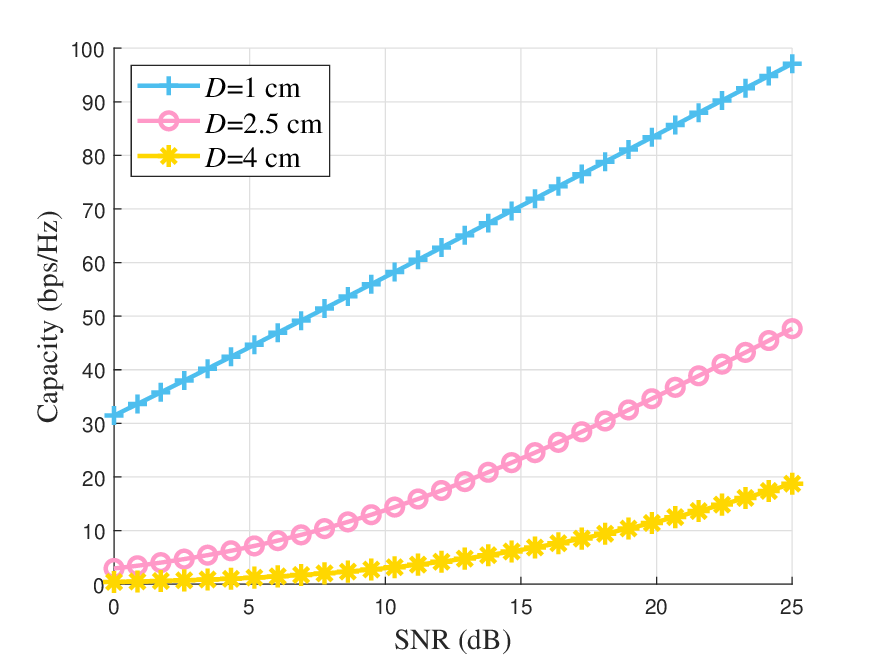}
\end{minipage}
}
\centering
\caption{The impact of transceiver coil numbers and transceiver distance on the OAM-NFC capacity.}
\end{figure}

\subsection{Impact of Transmit and Receive Coil Numbers on Capacity}
Figure~\ref{fig:capacity_N_M} plots the OAM-NFC capacity simulation results of different numbers of transmit and receive coils at $13.56$ MHz with channel estimation. Other variables remain the same as those in Section \ref{sec:HFSS_model}. We compare OAM-NFC capacities for $N_t\hspace{-0.1cm}=\hspace{-0.1cm}N_r\hspace{-0.1cm}=\hspace{-0.1cm}2$, $N_t\hspace{-0.1cm}=\hspace{-0.1cm}N_r\hspace{-0.1cm}=\hspace{-0.1cm}4$, $N_t\hspace{-0.1cm}=\hspace{-0.1cm}4, N_r\hspace{-0.1cm}=\hspace{-0.1cm}8$, and $N_t\hspace{-0.1cm}=\hspace{-0.1cm}N_r\hspace{-0.1cm}=\hspace{-0.1cm}8$ with SNR varying from $0$ to $25$ dB. Simulation results in Fig.~\ref{fig:capacity_N_M} are consistent with the analyses in Section \ref{sec:Capacity}. Moreover, for $N_t\hspace{-0.1cm}=\hspace{-0.1cm}4$, increasing $N_r$ from $4$ to $8$ does not increase the capacity significantly. Therefore, a large number of $N_r$ is unnecessary for increasing the OAM-NFC capacity when $N_t$ is small.

\subsection{Impact of Transceiver Distance on Capacity}
Figure~\ref{fig:capacity_D} plots the OAM-NFC capacity simulation results of different transceiver distances at $13.56$ MHz with channel estimation. Other variables remain the same as those in Section \ref{sec:HFSS_model}. In Fig.~\ref{fig:Capacity_theo_D}, we compare OAM-NFC capacities for $D\hspace{-0.1cm}=\hspace{-0.1cm}1$, $2.5$, and $4$ cm with SNR varying from $0$ to $50$ dB. Simulation results in Fig.~\ref{fig:capacity_D} are consistent with the analyses in Section \ref{sec:Capacity}. 
Fig.~\ref{fig:capacity_D} verifies that it is better to use our proposed OAM-NFC in near-distance scenarios than far-distance scenarios.

\begin{figure}[htbp]
\centering
\subfigure[OAM-NFC capacities with different transmit and receive coil ring radii.]{
\begin{minipage}{1\linewidth}\label{fig:Capacity_HFSS_R}
\centering
\includegraphics[scale=0.53]{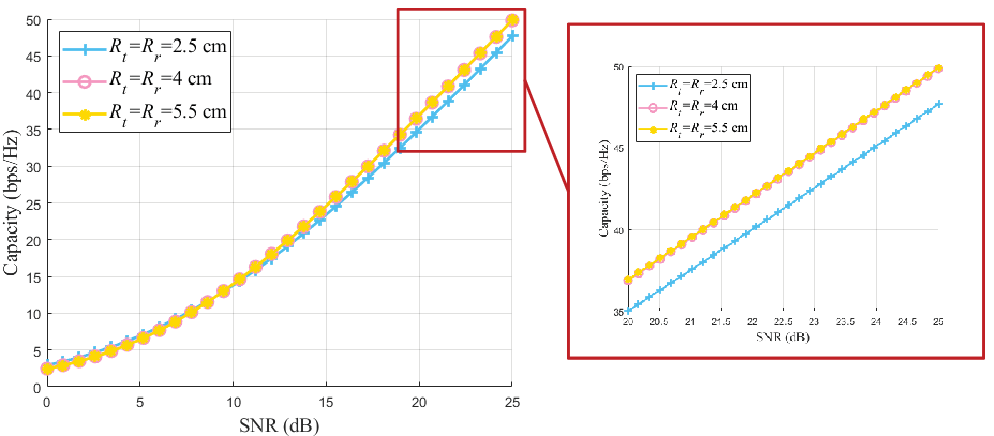}
\end{minipage}
}\\
\subfigure[OAM-NFC capacities with different transmit and receive coil radii.]{
\begin{minipage}{1\linewidth}\label{fig:Capacity_HFSS_rr}
\centering
\includegraphics[scale=0.6]{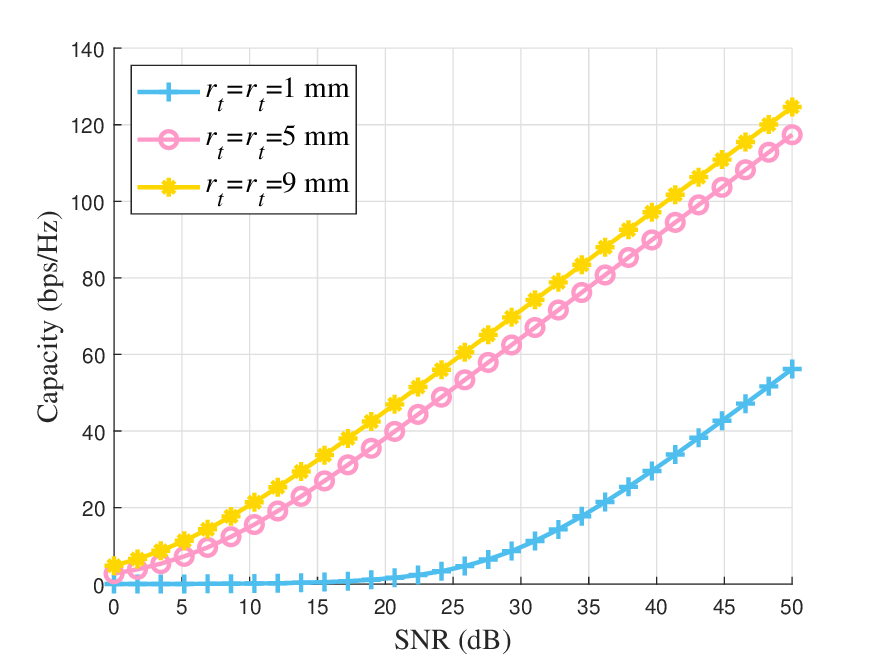}
\end{minipage}
}
\centering
\caption{The impact of transceiver ring radii and coil radii on the OAM-NFC capacity.}
\end{figure}

\subsection{Impact of Transmit and Receive Coil Ring Radii on Capacity}
Figure~\ref{fig:Capacity_HFSS_R} plots the OAM-NFC capacity simulation results of different transmit and receive coil ring radii at $13.56$ MHz with channel estimation. Other variables remain the same as those in Section \ref{sec:HFSS_model}. In Fig.~\ref{fig:Capacity_HFSS_R}, we compare OAM-NFC capacities with $R_t\hspace{-0.1cm}=\hspace{-0.1cm}R_r\hspace{-0.1cm}=\hspace{-0.1cm}2.5$, $4$ cm, and $5.5$ cm. Capacities of different coil ring radii are almost the same for SNR $<20$ dB. For SNR $\ge 20$ dB, OAM-NFC with $R_t\hspace{-0.1cm}=\hspace{-0.1cm}R_r\hspace{-0.1cm}=\hspace{-0.1cm}2.5$ cm has the minimum capacity while OAM-NFC with $R_t\hspace{-0.1cm}=\hspace{-0.1cm}R_r\hspace{-0.1cm}=\hspace{-0.1cm}4$ cm has the maximum capacity. This is because larger $R_t$ and $R_r$ lead to weaker mutual inductance between each two transmit coils, which is consistent with the analyses in Section \ref{sec:Capacity}. The impact of transmit and receive coil ring radii is small and our proposed OAM-NFC can get a high capacity even with relatively small coil rings.

\subsection{Impact of Transmit and Receive Coil Radii on Capacity}
Figure~\ref{fig:Capacity_HFSS_rr} plots the OAM-NFC capacity simulation results of different transmit and receive coil radii at $13.56$ MHz with channel estimation. Other variables remain the same as those in Section \ref{sec:HFSS_model}. In Fig.~\ref{fig:Capacity_HFSS_rr}, we compare OAM-NFC capacities with $r_t\hspace{-0.1cm}=\hspace{-0.1cm}r_r\hspace{-0.1cm}=\hspace{-0.1cm}1$, $5$, and $9$ mm. Simulation results in Fig.~\ref{fig:Capacity_HFSS_rr} are consistent with the analyses in Section \ref{sec:Capacity}. Since the capacity of $r_t\hspace{-0.1cm}=\hspace{-0.1cm}r_r\hspace{-0.1cm}=\hspace{-0.1cm}9$ mm is not significantly higher than that of $r_t\hspace{-0.1cm}=\hspace{-0.1cm}r_r\hspace{-0.1cm}=\hspace{-0.1cm}5$ mm, it is reasonable to adopt relatively small coils for saving space while still providing a high capacity.

\section{Conclusions}\label{sec:Conclusion}
In this paper, we modeled and analyzed the performance of the high-capacity OAM-NFC system and validated its feasibility as well as capacity enhancement. With coils circularly equipped at the transmitter and receiver, OAM signals can be transmitted, received, and detected. Based on the OAM-NFC mutual inductance channel model, we developed the OAM-NFC generation and detection schemes. Then, we derived and analyzed our proposed OAM-NFC capacity as well as those of SISO-NFC and MIMO-NFC. After that, we validated the feasibility and capacity enhancement of our proposed OAM-NFC with simulations. The impacts of the transceiver misalignment, transceiver coil numbers, transceiver distance, transceiver coil ring radii, and coils radii on the capacity of OAM-NFC were also analyzed. Simulation results show that our proposed OAM-NFC can be applied in NFC scenarios and increase the channel capacity significantly.

\begin{appendices}
\section{Proof for Theorem~\ref{the:channel_matrix}}\label{pro:multi-coil channel}
$M_{m,n}$ in Eq.~\eqref{eq:M} can be simplified as follows\cite{misalign_coils_M44,misalign_coils_M}:
\begin{align}
M_{m,n}=\frac{2\mu_0K_tK_r\sqrt{r_tr_r}}{\pi}\int_0^\pi \frac{\left(\cos\theta-\frac{d_{m,n}}{ r_r}\cos\phi\right)\Psi(k_{m,n})}{k_{m,n}\sqrt{V_{m,n}^3}}d\phi,
\label{eq:Mtr}
\end{align}
where
\begin{align}
\left\{\!\!
\begin{array}{ll}
d_{m,n}\hspace{-0.1cm}=\hspace{-0.1cm}\sqrt{\left(R_t \cos\frac{2\pi n}{N_t}-x_m\right)^2+\left(R_t \sin\frac{2\pi n}{N_t}-y_m\right)^2};\\
\\
V_{m,n}\hspace{-0.1cm}=\hspace{-0.1cm}\sqrt{1\hspace{-0.1cm}+\hspace{-0.1cm}\left({d_{m,n}}/{r_r}\right)^2\hspace{-0.1cm}-\hspace{-0.1cm}2\left({d_{m,n}}/{r_r}\right)\cos\phi\cos\theta\hspace{-0.1cm}-\hspace{-0.1cm}\cos^2\phi\sin^2\theta};\\
\\
k_{m,n}\hspace{-0.1cm}=\hspace{-0.1cm}2\sqrt{\frac{r_tr_rV_{m,n}}{\left(r_t+r_rV_{m,n}\right)^2+\left(z_m-r_r\cos\phi\sin\theta\right)^2}},
\end{array}
\label{eq:Mmnd}
\right.
\end{align}
and
\begin{align}
\left\{\!\!
\begin{array}{ll}
\Psi(\xi)=\left(1-\frac{\xi^2}{2}\right)K(\xi)-E(\xi);\\
\\
K(\xi)=\int_0^\pi\frac{1}{\sqrt{1-\xi^2\sin^2t}}dt;\\
\\
E(\xi)=\int_0^\pi\sqrt{1-\xi^2\sin^2t}dt.
\end{array}
\right.
\end{align}
In Eq.~\eqref{eq:Mmnd}, $(x_m,y_m,z_m)$ denotes the coordinate of the $m$th receive coil. For $\theta_x\hspace{-0.1cm}=\hspace{-0.1cm}\theta_y\hspace{-0.1cm}=\hspace{-0.1cm}0$, $(x_m,y_m,z_m)$ can be given as follows:
\begin{align}
\left\{\!\!
\begin{array}{ll}
x_m=d_x+R_r\cos\frac{2\pi m}{N_r};\\
y_m=d_y+R_r\sin\frac{2\pi m}{N_r};\\
z_m=D.
\end{array}
\right.
\end{align}
For non-zero $\theta_x$ or $\theta_y$, $(x_m,y_m,z_m)$ can be given as follows:
\begin{align}
&\left\{\!\!
\begin{array}{ll}
x_m\hspace{-0.1cm}=\hspace{-0.1cm}d_x\hspace{-0.1cm}+\hspace{-0.1cm}\left(\frac{\tan\theta_y}{\tan\theta}\right)R_r\cos\frac{2\pi m}{N_r}\hspace{-0.1cm}+\hspace{-0.1cm}\left(\frac{\tan\theta_x}{\tan\theta\sqrt{\tan^2\theta+1}}\right)R_r\sin\frac{2\pi m}{N_r};\\
y_m\hspace{-0.1cm}=\hspace{-0.1cm}d_y\hspace{-0.1cm}-\hspace{-0.1cm}\left(\frac{\tan\theta_x}{\tan\theta}\right)R_r\cos\frac{2\pi m}{N_r}\hspace{-0.1cm}+\hspace{-0.1cm}\left(\frac{\tan\theta_y}{\tan\theta\sqrt{\tan^2\theta+1}}\right)R_r\sin\frac{2\pi m}{N_r};\\
z_m\hspace{-0.1cm}=\hspace{-0.1cm}D\hspace{-0.1cm}-\hspace{-0.1cm}\left(\frac{\tan\theta}{\sqrt{\tan^2\theta+1}}\right)R_r\sin\frac{2\pi m}{N_r}.
\end{array}
\right.
\end{align}
Similarly to Eq.~\eqref{eq:Mtr}, $M^t_{n_1,n_2}$ in Eq.~\eqref{eq:Mt} with $n_1\hspace{-0.1cm}\ne\hspace{-0.1cm} n_2$ can be simplified as follows:
\begin{align}
M^t_{n_1,n_2}=\frac{2\mu_0K^2_tr_t}{\pi}\int_0^\pi \frac{\left(1-\frac{d^t_{n_1,n_2}}{ r_t}\cos\phi\right)\Psi(k^t_{n_1,n_2})}{k^t_{n_1,n_2}\sqrt{(V^t_{n_1,n_2})^3}}d\phi,
\label{eq:Mt2}
\end{align}
where
\begin{align}
\left\{\!\!
\begin{array}{ll}
d^t_{n_1,n_2}=R_t\sqrt{2-2\cos\left[\frac{2\pi}{N_t}\left(n_1-n_2\right)\right]};\\
\\
V^t_{n_1,n_2}=\sqrt{1+\left({d^t_{n_1,n_2}}/{r_r}\right)^2-2\left({d^t_{n_1,n_2}}/{r_r}\right)\cos\phi};\\
\\
k^t_{n_1,n_2}=2{\sqrt{V^t_{n_1,n_2}}}/{(1+V^t_{n_1,n_2})}.
\end{array}
\right.
\label{eq:MtdVk}
\end{align}

Based on the mutual inductances given above, we analyze the induced voltages across each transmit and receive coils. First, we analyze the induced voltage across each transmit coil, which is generated by other transmit coils. Let $v^t_{n_2}$ denote the source voltage across the $n_2$th transmit coil. The induced voltage generated by the $n_2$th transmit coil across the $n_1$th transmit coil, denoted by $v^s_{n_1,n_2}$, is derived as $v_{n_1,n_2}^s \hspace{-0.1cm}=\hspace{-0.1cm} \frac{-j\omega M^t_{n_1,n_2}}{Z_t}v^t_{n_2}$. Then, the voltage across the $n$th transmit coil, denoted by $v_n$, can be given as follows:
\begin{align}
v_n=v^t_n+\hspace{-0.2cm}\sum_{i=0,i\ne n}^{N-1}\hspace{-0.2cm}v^s_{n,i}=v^t_n+\hspace{-0.2cm}\sum_{i=0,i\ne n}^{N-1}\hspace{-0.2cm} \frac{-j\omega M^t_{n,i}}{Z_t}v^t_i.
\end{align}
Similarly, we give the induced voltage across each receive coil generated by all the transmit coils. We denote by $v_{m,n} \hspace{-0.1cm}=\hspace{-0.1cm} \frac{-j\omega M_{m,n}}{Z_t}v_n$ the induced voltage generated by the $n$th transmit coil across the $m$th receive coil. Therefore, the total induced voltage across the $m$th receive coil can be derived as follows:
\begin{align}
v^r_m&=\sum_{n=0}^{N-1}\frac{-j\omega M_{m,n}}{Z_t}v_n
\nonumber\\&=\sum_{n=0}^{N-1}\frac{-j\omega M_{m,n}}{Z_t}\left(v^t_n+\sum_{i=0,i\ne n}^{N-1} \frac{-j\omega M^t_{n,i}}{Z_t}v^t_i\right)
\nonumber\\&=\frac{-j\omega }{Z_t}\sum_{n=0}^{N-1}M_{m,n}v^t_n-\frac{\omega^2}{Z_t^{2}}\sum_{n=0}^{N-1}M_{m,n}\hspace{-0.2cm}\sum_{i=0,i\ne n}^{N-1}\hspace{-0.2cm}M^t_{n,i}v^t_i.
\label{eq:induced_V}
\end{align}
Eq.~\eqref{eq:induced_V} can be expressed in matrix form as follows:
\begin{align}
\boldsymbol{\mathrm v}^r=\frac{-j\omega}{Z_t}\boldsymbol{\mathrm M}\boldsymbol{\mathrm v}^t-\frac{\omega^2}{Z_t^{2}}\boldsymbol{\mathrm M}\boldsymbol{\mathrm M}^t\boldsymbol{\mathrm v}^t,
\label{eq:vr_vt}
\end{align}
where the main diagonal elements of $\boldsymbol{\mathrm M}^t$ are set as $0$. Thus, $\boldsymbol{\mathrm H}$ can be given as follows:
\begin{align}
\boldsymbol{\mathrm H}=\frac{-j\omega}{Z_t}\boldsymbol{\mathrm M}-\frac{\omega^2}{Z_t^{2}}\boldsymbol{\mathrm M}\boldsymbol{\mathrm M}^t.
\end{align}

\section{Proof for Theorem~\ref{the:circulant_matrix}}\label{pro:circulant channel}
First, we prove that $\boldsymbol{\mathrm M}^t$ is a circulant matrix. Based on Eq.~\eqref{eq:MtdVk}, since $d^t_{n_1,n_2}\hspace{-0.1cm}=\hspace{-0.1cm}d^t_{n_1+1,n_2+1}$, we can obtain $V^t_{n_1,n_2}\hspace{-0.1cm}=\hspace{-0.1cm}V^t_{n_1+1,n_2+1}$ and $k^t_{n_1,n_2}\hspace{-0.1cm}=\hspace{-0.1cm}k^t_{n_1+1,n_2+1}$. Thus, $\boldsymbol{\mathrm M}^t$ is a circulant matrix, given as follows:
\small{
\begin{align}
\boldsymbol{\mathrm M}^t\hspace{-0.1cm}=\hspace{-0.1cm}
\begin{bmatrix}
M^t_{1,1} & M^t_{1,2} & \cdots & M^t_{1,N_t}\\
M^t_{1,N_t} & M^t_{1,1} & \cdots & M^t_{1,N_t-1}\\
\vdots & \vdots & \ddots & \vdots \\
M^t_{1,2}  & M^t_{1,3} & \cdots & M^t_{1,1}
\end{bmatrix}.
\label{eq:Mnnb}
\end{align}
}

\begin{figure*}[ht]
\begin{align}
\setcounter{equation}{50}
\widehat{C}_{\rm{OAM}}\hspace{-0.1cm}&\le\hspace{-0.1cm}\sum_{l=0}^{N_t-1}{\rm{log}}\hspace{-0.1cm}\left(1\hspace{-0.1cm}+\hspace{-0.1cm}\frac{P_tI}{N_0N_t}\frac{4\mu^2_0\omega^2K_t^2K_r^2r_tr_r}{\pi^2|Z_t|^2N_r}\left|N_r\sqrt{\sum^{N_r}_{m=1}\frac{\left|e^{-jl\frac{2\pi m}{N_r}}\hspace{-0.1cm}\int_0^\pi \frac{\left(1-d_{m,1}/r_t\right)\cos\phi\Psi(\widehat k_{m,1})}{\widehat k_{m,1}\sqrt{\widehat V_{m,1}^3}}d\phi\right|^2}{N_r}}\right|^2\right)\nonumber\\
\hspace{-0.1cm}&=\hspace{-0.1cm}\sum_{l=0}^{N_t-1}{\rm{log}}\hspace{-0.1cm}\left(1\hspace{-0.1cm}+\hspace{-0.1cm}\frac{P_tI}{N_0N_t}\frac{4\mu^2_0\omega^2K_t^2K_r^2r_tr_r}{\pi^2|Z_t|^2}\sum^{N_r}_{m=1}\left|e^{-jl\frac{2\pi m}{N_r}}\hspace{-0.1cm}\int_0^\pi \frac{\left(1-d_{m,1}/r_t\right)\cos\phi\Psi(\widehat k_{m,1})}{\widehat k_{m,1}\sqrt{\widehat V_{m,1}^3}}d\phi\right|^2\right)\nonumber\\
\hspace{-0.1cm}&=\hspace{-0.1cm}N_t{\rm{log}}\hspace{-0.1cm}\left(1\hspace{-0.1cm}+\hspace{-0.1cm}\frac{P_tI}{N_0N_t}\frac{4\mu^2_0\omega^2K_t^2K_r^2r_tr_r}{\pi^2|Z_t|^2}\sum^{N_r}_{m=1}\left|\int_0^\pi \frac{\left(1-d_{m,1}/r_t\right)\cos\phi\Psi(\widehat k_{m,1})}{\widehat k_{m,1}\sqrt{\widehat V_{m,1}^3}}d\phi\right|^2\right).
\label{eq:COAM2_upper_bound}
\end{align}
\hrulefill
\end{figure*}
\begin{figure*}[ht]
\begin{align}
\setcounter{equation}{52}
\widehat{C}_{\rm{OAM}}\hspace{-0.1cm}&\ge\hspace{-0.1cm}\sum_{l=0}^{N_t-1}{\rm{log}}\hspace{-0.1cm}\left(1\hspace{-0.1cm}+\hspace{-0.1cm}\frac{P_tI}{N_0N_t}\frac{4\mu^2_0\omega^2K_t^2K_r^2r_tr_r}{\pi^2|Z_t|^2N_r}\left|N_r\sqrt[N_r]{\prod^{N_r}_{m=1}e^{-jl\frac{2\pi m}{N_r}}\hspace{-0.1cm}\int_0^\pi \frac{\left(1-d_{m,1}/r_t\right)\cos\phi\Psi(\widehat k_{m,1})}{\widehat k_{m,1}\sqrt{\widehat V_{m,1}^3}}d\phi}\right|^2\right)\nonumber\\
\hspace{-0.1cm}&=\hspace{-0.1cm}\sum_{l=0}^{N_t-1}{\rm{log}}\hspace{-0.1cm}\left(1\hspace{-0.1cm}+\hspace{-0.1cm}\frac{P_tI}{N_0N_t}\frac{4\mu^2_0\omega^2K_t^2K_r^2r_tr_r}{\pi^2|Z_t|^2}N_r\left|\prod^{N_r}_{m=1}e^{-jl\frac{2\pi m}{N_r}}\hspace{-0.1cm}\sqrt[N_r]{\int_0^\pi \frac{\left(1-d_{m,1}/r_t\right)\cos\phi\Psi(\widehat k_{m,1})}{\widehat k_{m,1}\sqrt{\widehat V_{m,1}^3}}d\phi}\right|^2\right)\nonumber\\
\hspace{-0.1cm}&=\hspace{-0.1cm}N_t{\rm{log}}\hspace{-0.1cm}\left(1\hspace{-0.1cm}+\hspace{-0.1cm}\frac{P_tI^2}{N_0}\frac{4\mu^2_0\omega^2K_t^2K_r^2r_tr_r}{\pi^2|Z_t|^2}\prod^{N_r}_{m=1}\left|\sqrt[N_r]{\int_0^\pi \frac{\left(1-d_{m,1}/r_t\right)\cos\phi\Psi(\widehat k_{m,1})}{\widehat k_{m,1}\sqrt{\widehat V_{m,1}^3}}d\phi}\right|^2\right).
\label{eq:COAM2_lower_bound}
\end{align}
\hrulefill
\end{figure*}

Then, we prove that $\boldsymbol{\mathrm M}$ is a block circulant matrix when the transmit coil ring and the receive coil ring are aligned with each other. For aligned transceivers, we can obtain $d_x\hspace{-0.1cm}=\hspace{-0.1cm}d_y\hspace{-0.1cm}=\hspace{-0.1cm}0$ and $\theta_x\hspace{-0.1cm}=\hspace{-0.1cm}\theta_y\hspace{-0.1cm}=\hspace{-0.1cm}0$. Thus, Eq.~\eqref{eq:Mtr} can be rewritten as follows:
\begin{align}
\setcounter{equation}{47}
M_{m,n}=\frac{2\mu_0K_tK_r\sqrt{r_tr_r}}{\pi}\int_0^\pi \frac{\left(1-\frac{d_{m,n}}{ r_r}\cos\phi\right)\Psi(k_{m,n})}{k_{m,n}\sqrt{V_{m,n}^3}}d\phi,
\label{eq:MmnA}
\end{align}
where $d_{m,n}$, $V_{m,n}$ and $k_{m,n}$ are simplified as follows:
\begin{align}
\left\{\!\!
\begin{array}{ll}
d_{m,n}=\sqrt{R_t^2+R_r^2-2R_tR_r\cos[\frac{2\pi}{N_t}\left(n-m/I\right)]};\\
\\
V_{m,n}=\sqrt{1+\left({d_{m,n}}/{r_r}\right)^2-2\left({d_{m,n}}/{r_r}\right)\cos\phi};\\
\\
k_{m,n}=2\sqrt{\frac{\left({r_r}/{r_t}\right)V_{m,n}}{\left[1+\left({r_r}/{r_t}\right)V_{m,n}\right]^2+\left[D/r_t\right]^2}}.
\end{array}
\right.
\end{align}
Since $d_{m,n}\hspace{-0.1cm}=\hspace{-0.1cm}d_{m+I,n+1}$, we can obtain $V_{m,n}\hspace{-0.1cm}=\hspace{-0.1cm}V_{m+I,n+1}$ and $k_{m,n}\hspace{-0.1cm}=\hspace{-0.1cm}k_{m+I,n+1}$. Therefore, $\boldsymbol{\mathrm M}$ is a block circulant matrix. Thus, based on Eq.~\eqref{eq:channel_matrix} associated with Eqs.~\eqref{eq:Mnnb}, $\boldsymbol{\mathrm H}$ is also a block circulant matrix given as follows:
\begin{align}
\boldsymbol{\mathrm H}&\hspace{-0.1cm}=\hspace{-0.15cm}
\begin{bmatrix}
H_{1,1} &\hspace{-0.15cm} H_{1,2} &\hspace{-0.15cm} \cdots &\hspace{-0.15cm} H_{1,N_t}\\
H_{2,1} &\hspace{-0.15cm} H_{2,2} &\hspace{-0.15cm} \cdots &\hspace{-0.15cm} H_{2,N_t}\\
\vdots & \vdots & \ddots & \vdots \\
H_{N_r,1}  &\hspace{-0.1cm} H_{N_r,2} &\hspace{-0.1cm} \cdots &\hspace{-0.1cm} H_{N_r,N_t}
\end{bmatrix}
\hspace{-0.15cm}=\hspace{-0.15cm}
\begin{bmatrix}
\boldsymbol{\mathrm h}_1 &\hspace{-0.1cm} \boldsymbol{\mathrm h}_2 &\hspace{-0.1cm} \cdots &\hspace{-0.1cm} \boldsymbol{\mathrm h}_{N_t}\\
\boldsymbol{\mathrm h}_{N_t} &\hspace{-0.1cm} \boldsymbol{\mathrm h}_1 &\hspace{-0.1cm} \cdots &\hspace{-0.1cm} \boldsymbol{\mathrm h}_{N_t-1}\\
\vdots &\hspace{-0.1cm} \vdots &\hspace{-0.1cm} \ddots &\hspace{-0.1cm} \vdots \\
\boldsymbol{\mathrm h}_2 &\hspace{-0.1cm} \boldsymbol{\mathrm h}_3 &\hspace{-0.1cm} \cdots &\hspace{-0.1cm} \boldsymbol{\mathrm h}_1
\end{bmatrix}.
\end{align}

\section{Proof for Theorem~\ref{the:upper_lower}}\label{pro:COAM2_upper_lower}
First, we prove Eq.~\eqref{eq:COAM2_upper}. Because $\frac{1}{N}\sum^N_{n=1}a_n\hspace{-0.1cm}\le\hspace{-0.1cm}\sqrt{\frac{\sum^N_{n=1}}{N}a_n}$, the upper bound of $\widehat{C}_{\rm{OAM}}$ can be derived as Eq.~\eqref{eq:COAM2_upper_bound}.
Because $\cos\frac{2\pi (1-1)}{N_r}\hspace{-0.1cm}=\hspace{-0.1cm}1\hspace{-0.1cm}\ge\hspace{-0.1cm}\cos\frac{2\pi (m-1)}{N_r}$, $d_{m,1}$ achieves its minimum value with $m\hspace{-0.1cm}=\hspace{-0.1cm}I$ and is given as $d_{m,1}\hspace{-0.1cm}\ge\hspace{-0.1cm} \sqrt{R_r^2+R_t^2-2R_rR_t} \hspace{-0.1cm}=\hspace{-0.1cm} \left|R_r-R_t\right|$. Since for $r_r\hspace{-0.1cm}=\hspace{-0.1cm}r_t$, $M_{m,n}$ decreases as $d_{m,n}$ increases. Then, replacing $d_{m,1}$ with $\left|R_r-R_t\right|$, Eq.~\eqref{eq:COAM2_upper_bound} can be further derived as follows:
\begin{align}
\setcounter{equation}{51}
\widehat{C}_{\rm{OAM}}\hspace{-0.1cm}&<\hspace{-0.1cm}N_t{\rm{log}}\Bigg(1\hspace{-0.1cm}+\hspace{-0.1cm}\frac{P_tI^2}{N_0}\frac{4\mu^2_0\omega^2K_t^2K_r^2r_tr_r}{\pi^2|Z_t|^2}
\nonumber\\&\hspace{1.5cm}\left|\int_0^\pi \frac{\left(1-|R_r\hspace{-0.1cm}-\hspace{-0.1cm}R_t|/r_t\right)\cos\phi\Psi(\widehat k^{upper})}{\widehat k^{upper}\sqrt{(\widehat V^{upper})^3}}d\phi\right|^2\Bigg),
\end{align}
which is the upper limit of $\widehat{C}_{\rm{OAM}}$. Thus, Eq.~\eqref{eq:COAM2_upper} follows.

Then, we prove Eq.~\eqref{eq:COAM2_lower}. Because $\frac{1}{N}\sum^N_{n=1}a_n\hspace{-0.1cm}\ge\hspace{-0.1cm}\sqrt[N]{\prod^{N}_{n=1}a_n}$, the lower bound of $\widehat{C}_{\rm{OAM}}$ can be derived as Eq.~\eqref{eq:COAM2_lower_bound}.
Therefore, replacing $d_{m,1}$ with $R_r\hspace{-0.1cm}+\hspace{-0.1cm}R_t$, Eq.~\eqref{eq:COAM2_lower_bound} can be further derived as follows:
\begin{align}
\setcounter{equation}{53}
\widehat{C}_{\rm{OAM}}\hspace{-0.1cm}&>\hspace{-0.1cm}N_t{\rm{log}}\Bigg(1\hspace{-0.1cm}+\hspace{-0.1cm}\frac{P_tI^2}{N_0}\frac{4\mu^2_0\omega^2K_t^2K_r^2r_tr_r}{\pi^2|Z_t|^2}
\nonumber\\&\hspace{1.5cm}\left|\int_0^\pi \frac{\left[1- \left(R_r+R_t\right)/r_t\right]\cos\phi\Psi(\widehat k^{lower})}{\widehat k^{lower}\sqrt{(\widehat V^{lower})^3}}d\phi\right|^2\Bigg),
\end{align}
which is the lower limit of $\widehat{C}_{\rm{OAM}}$. Thus, Eq.~\eqref{eq:COAM2_lower} follows.
\end{appendices}

\bibliographystyle{IEEEtran}
\bibliography{References}

\ifCLASSOPTIONcaptionsoff
  \newpage
\fi

\end{document}